\RequirePackage{fixltx2e}
\documentclass[11pt]{article}  
\usepackage[utf8]{inputenc}
\usepackage{amsmath, amsfonts, amsthm, amssymb, amscd,a4wide}
\usepackage[mathcal]{euscript} 
\usepackage{dsfont, pxfonts, mathrsfs, accents, verbatim}
\usepackage{pdfpages}
\input{epsf}
\usepackage{graphicx}
\usepackage{xcolor}
\usepackage{tikz} 
\usepackage{fancyhdr}
 
\usepackage{microtype}
\usepackage{slashed}
\usepackage{xcolor}

\usepackage{hyperref,footnotebackref}
\hypersetup{linktoc = all}                
\hypersetup{unicode}
\hypersetup{hidelinks}
\hypersetup{bookmarksnumbered}
\newcommand{\abs}[1]{\lvert#1\rvert}
\theoremstyle{plain}
        \newtheorem{theorem}{Theorem}[section]
\newtheorem{definition}[theorem]{Definition}
        \newtheorem{proposition}[theorem]{Proposition}

        \newtheorem{corollary}[theorem]{Corollary}
        \newtheorem{remark}[theorem]{Remark}

\numberwithin{equation}{section}
\DeclareMathAlphabet\mathbfcal{OMS}{cmsy}{b}{n}
\usepackage[most]{tcolorbox} 
\newtcolorbox{bothok}{%
     enhanced, breakable, size=minimal, parbox=false, after={\par}, 
     before upper={\indent}, colback=white, 
     overlay = {\draw[line width=2pt] (frame.north east) -|
                       ([xshift=3mm]frame.east)|-(frame.south east);},
     overlay first={\draw[line width=2pt] (frame.north east) -|
                           ([xshift=3mm]frame.south east);},
     overlay middle={\draw[line width=2pt] ([xshift=3mm]frame.north east) -- 
                              ([xshift=3mm]frame.south east);},
     overlay last={\draw[line width=2pt] ([xshift=3mm]frame.north east)|-
                          (frame.south east);},
}
\newtcolorbox{blfok}{%
     enhanced, breakable, size=minimal, parbox=false, after={\par}, 
     before upper={\indent}, colback=white, 
     overlay = {\draw[dotted, line width=2pt] (frame.north east) -|
                       ([xshift=3mm]frame.east)|-(frame.south east);},
     overlay first={\draw[dotted, line width=2pt] (frame.north east) -|
                           ([xshift=3mm]frame.south east);},
     overlay middle={\draw[dotted, line width=2pt] ([xshift=3mm]frame.north east) -- 
                              ([xshift=3mm]frame.south east);},
     overlay last={\draw[dotted, line width=2pt] ([xshift=3mm]frame.north east)|-
                          (frame.south east);},
}
\newtcolorbox{plfok}{%
     enhanced, breakable, size=minimal, parbox=false, after={\par}, 
     before upper={\indent}, colback=white, 
     overlay = {\draw[dashed, line width=2pt] (frame.north east) -|
                       ([xshift=3mm]frame.east)|-(frame.south east);},
     overlay first={\draw[dashed, line width=2pt] (frame.north east) -|
                           ([xshift=3mm]frame.south east);},
     overlay middle={\draw[dashed, line width=2pt] ([xshift=3mm]frame.north east) -- 
                              ([xshift=3mm]frame.south east);},
     overlay last={\draw[dashed, line width=2pt] ([xshift=3mm]frame.north east)|-
                          (frame.south east);},
}

\newcommand \brho {\underline \rho}
\newcommand \bv {\underline v}
\newcommand \bSigma {\underline \Sigma}

\newcommand \barU {\underline U}
\newcommand \barA {\underline A}
\newcommand \bara {\underline a} 
\newcommand \barnu {\underline \nu} 

\newcommand \Psharp {P^\sharp}
\newcommand \Qsharp {Q^\sharp}

\newcommand \et {\widetilde e}

\newcommand \la 	\langle
\newcommand \ra 	\rangle 
\newcommand \barlambda {\underline \lambda}
\newcommand \barmu {\underline \mu}

\newcommand \dbf 	{{\mathbf d}}

\newcommand \Ncal 	{\mathcal N}
\newcommand \Vcal 	{\mathcal V}
\newcommand \Vcalt 	{\widetilde{\mathcal V}}  

\newcommand \Gcal 	{\mathcal G}
\newcommand \Ecalt 	{\mathcal {\widetilde E}}

\newcommand \Pn  	{P^n}
\newcommand \Qn  	{Q^n}
\newcommand \Pinfty  {P^\infty}
\newcommand \Qinfty  {Q^\infty}

\newcommand \barP 	{\underline P}
\newcommand \barQ 	{\underline Q} 
\newcommand \bei 	{\begin{itemize}}
\newcommand \eei 	{\end{itemize}}
\newcommand \del	 \partial
\newcommand \auth 	{\textsc}   
   
\newcommand \Hcal 	{\mathcal H}   
\newcommand \Lcal 	{\mathcal L}   
\newcommand \Ecal 	{\mathcal E}

\newcommand \loc   	{\textnormal{loc}} 
\newcommand \RR 	{\mathbb R}   
\newcommand \eps 	\epsilon  
\newcommand \be 	{\begin{equation}}
\newcommand \ee 	{\end{equation}} 
\newcommand \bel 	{\be \label}
\let\oldmarginpar\marginpar
\renewcommand\marginpar[1]{\ifhmode\unskip\fi\- \oldmarginpar[\raggedleft\footnotesize #1]%
{\raggedright\footnotesize #1}}
\newcommand \lam 	\lambda

\newcommand \tmax  {t_\text{max}} 
\newcommand \tmin  {t_\text{min}}

\newcommand{\WARNONCE}{\gdef\WARNONCE{}\GenericError{}{There are still BLF/PLF comments in this file ! \space}{}{}}
\newcommand \Sb 	{\overline S}
\newcommand \Pb 	{\overline P}
\newcommand \psib       {\overline \psi}
\newcommand \lambdab    {\overline \lambda}
\renewcommand \th 	\theta 
\newcommand \sign 	{\operatorname{sign}}

\renewcommand \geq \geqslant
\renewcommand \leq \leqslant

\newcommand \bea  {\begin{eqnarray}}
\newcommand \eea  {\end{eqnarray}}

\newcommand \lbrac \llbracket 
\newcommand \rbrac \rrbracket 
\newcommand \R    {\mathbb{R}}

\newcommand \HoneWoneone {\texorpdfstring{$H^1$-$W^{1,1}$}{H\textonesuperior\ - W\textonesuperior\ \textonesuperior}}
\begin{document}

\title{On the global evolution of self-gravitating matter.
\\
Nonlinear interactions in Gowdy symmetry 
\footnotetext{$^1$ Princeton Center for Theoretical Science (PCTS), Princeton University, Princeton, New Jersey 08544, USA. Email: {\tt blefloch@princeton.edu}.
\\
$^2$ Laboratoire Jacques-Louis Lions, Centre National de la Recherche Scientifique, 
Sorbonne Universit\'e,  
4 Place Jussieu, 75252 Paris, France. Email: {\tt contact@philippelefloch.org}. 
\\
{\it Key Words.} Hyperbolic balance laws; relativistic Euler equations; Gowdy symmetry; energy functional; nonlinear stability. 
\textit{\ AMS Class.}    35L65, 76L05,  83C05.
\\
Completed on: May 25, 2018. To appear in: Archive for Rational Mechanics and Analysis. 
}} 
\author{Bruno Le Floch$^1$ and Philippe G. LeFloch$^2$   
}

\date{} 

\maketitle

\begin{abstract}
We are interested in the evolution of a compressible fluid under its self-generated gravitational field.  Assuming here Gowdy symmetry, we investigate the algebraic structure of the Euler equations satisfied by the mass density and velocity field. We exhibit several interaction functionals that provide us with a uniform control on weak solutions in suitable Sobolev norms or in bounded variation.  These functionals allow us to study the local regularity and nonlinear stability properties of weakly regular fluid flows governed by the Euler-Gowdy system. In particular for the Gowdy equations, 
we prove that a spurious matter field arises under weak convergence, and we establish the nonlinear stability of weak solutions.
\end{abstract}

\vfill

\setcounter{secnumdepth}{6}
\setcounter{tocdepth}{1} 
\hypersetup{bookmarksdepth=2}

\tableofcontents  

\

\

\newpage 


\section{Introduction}
\label{sec:intro}

This is the first of a series of papers \cite{LeFlochLeFloch-2, LeFlochLeFloch-3}, where we study the global dynamics of matter fields, especially the evolution of perfect compressible fluids, when self-gravitating effects are taken into account. In particular, we are interested in weak solutions to the relativistic Euler equations of continuum physics and in  developing mathematical techniques relevant for a variety of astrophysical and cosmological problems. We intend to bridge together methods of nonlinear analysis ---that have been introduced for the study of nonlinear hyperbolic conservation laws--- and mathematical methods of geometric analysis relevant for gravitational problems. 

The Euler equations for compressible fluids, especially in the relativistic setting under consideration, are a mathematically challenging system of \textsl{nonlinear hyperbolic conservation laws}\footnote{actually, a system of balance laws}, which accurately describes a broad set of complex wave phenomena observed in fluids. 
Progress on the mathematical analysis of solutions of the Euler equations, especially weak solutions with shock waves, often relies on making first 
certain symmetry assumptions, such as planar or spherical symmetry. As we will show below, the equations we consider in the present paper, while being significantly more involved, are similar to the equations of a compressible fluid evolving in a nozzle with variable cross-section. 

Hence, we assume here a symmetry, specifically Gowdy symmetry \cite{Gowdy-1974,Moncrief-1981}, on the self-gravitating fluid flow of interest, and in Section~\ref{sec:15} we introduce the \textsl{Euler-Gowdy system}, as we call it.
Having in mind an audience with a background on nonlinear hyperbolic conservation laws ---as described in Dafermos' reference textbook \cite{Dafermos-book} as well as \cite{Lax-1957,Glimm-1965,Lax-1971,LeFloch-book}--- we intentionally postpone most of the geometry-related material to a follow-up paper.  

We focus here on a major issue, i.e.~the derivation of \textsl{a priori~bounds} for weakly regular solutions to the Euler-Gowdy system. Namely, we  exhibit several energy-type functionals that allow us to investigate the stability of weak solutions under strong or weak convergence. These functionals provide a uniform control on weak solutions with Sobolev or bounded variation regularity. Interestingly, we discover here that, even when the matter density vanishes identically, a spurious matter field may arise under weak convergence. The follow-up paper \cite{LeFlochLeFloch-2} will rely on these results and establish a mathematical theory for the initial value problem associated with the Euler-Gowdy system. 

Recall that all classical works on the Gowdy equations (for instance \cite{Moncrief-1981}) assumed the metric to be of class $C^2$, and it is only in recent years that the study at a weak regularity level was studied in \cite{Barnes-2004,GrubicLeFloch-2013, GrubicLeFloch-2015,LeFlochRendall-2011, LeFlochSmulevici-2016,LeFlochStewart-2011}. We build here on this earlier study while providing significantly new insights on the regularity, integrability, and nonlinear stability properties of Gowdy-symmetric solutions. We recall that the construction of weak solutions (with bounded variation) to the equations of self-gravitating matter was achieved first by Christodoulou \cite{Christo-1992,Christo-1995} for a rather different matter model (i.e. a scalar field) and under the assumption of spherical symmetry.  We also refer to Rendall \cite{Rendall-book} for further related background. 

This paper is organized as follows.
In Section~\ref{sec:15} we present the Euler-Gowdy system together with two energy-like functionals that provide a control of the fluid variables as well as the geometric variables. 
In Section~\ref{sec:3} we investigate several notions of weak solutions and functionals for an ``essential'' nonlinear wave system extracted from the full model;
we observe that these essential equations only involve the fluid density and can be analyzed first independently from the fluid equations.  
We return in Section~\ref{sec:4} to the Euler-Gowdy system:
first, by considering spatially homogeneous solutions we observe that (even weak) solutions are not necessarily defined for ``all'' times; then we present our functional for the Euler-Gowdy system and the corresponding a priori bound for weak solutions. In Section~\ref{sec:5} we show that our functional can be used to control a suitable distance between two Gowdy solutions (whereas this issue for fluids is more involved and is postponed to \cite{LeFlochLeFloch-2}).
We end our investigation with several alternative classes of (Sobolev, bounded variation) regularity in Section~\ref{sec:6}.
For completeness we summarize the derivation of the Euler-Gowdy model in Appendix~\ref{sec:7}.


\section{The Euler-Gowdy system}
\label{sec:15} 

\subsection{The nonlinear hyperbolic system of interest}

With the above motivations in mind, we assume Gowdy symmetry and detail here the resulting \textsl{Euler-Gowdy system} (cf.~\cite{LeFlochRendall-2011} or Appendix~\ref{sec:7} for a derivation).
The unknowns are fluid variables $(\rho,v)$ and geometric variables $(U,A,a,\nu)$ depending upon a time coordinate denoted by $t \neq 0$ 
and a space coordinate $\theta \in T^1 = [0, 1]$ (the one-torus, i.e. with the standard periodic boundary condition).  
We denote partial derivatives with respect to the variables $t$ and $\theta$ using subscripts: $(\ )_t$ and $(\ )_\theta$.
The speed of sound $k\in[0,1]$, normalized by the speed of light, is a prescribed constant.

The \textsl{Euler equations} for a relativistic fluid with mass density $\rho=\rho(t,\theta) \geq 0$ and normalized velocity $v=v(t,\theta) \in (-1,1)$ read
\bel{fluid-II}
\aligned
& \Bigg( t  \, {\rho \over a} \, \frac{1+k^2v^2}{1 - v^2}\Bigg)_t
+ \Bigg( t \, \rho v \, \frac{1+ k^2}{1 - v^2} \Bigg)_\theta
=   
{\rho \over a} \, \Big( - k^2 + (1-k^2) \, \Sigma_0 \Big), 
\\
& \Bigg( t \, {\rho v \over a} \, \frac{1+k^2}{1 - v^2}\Bigg)_t 
+ \Bigg(t \, \rho \, \frac{k^2 + v^2}{1 - v^2} \Bigg)_\theta
= - \rho \, (1-k^2) \, \Sigma_1,
\endaligned
\ee
in which
\be
\aligned 
\Sigma_0  = \Sigma_0[t, U, A, a]  &\coloneqq   - t \, U_t + t^2 \, \big( U_t^2 + a^2U_\theta^2 \big) + {1 \over 4} e^{4U} (A_t^2 + a^2A_\theta^2),
\\
\Sigma_1 = \Sigma_1[t, U, A, a] &\coloneqq - t \, U_\theta + 2t^2 \, U_t \, U_\theta + {1 \over 2} e^{4U} \, A_t A_\theta.
\endaligned
\ee

On the other hand, the unknowns $U=U(t, \theta) \in \RR$ and $A = A(t, \theta) \in \RR$
 are determined from the 
second-order nonlinear wave equations ---referred to below as the \textsl{essential Gowdy equations}--- 
\bel{evolution-I}
\aligned
\Bigg( t \, {U_t \over a} - {1 \over 2 a} \Bigg)_t  
- \big( t \, a U_\theta \big)_\theta 
& = {1 \over 2ta} e^{4U} \, \big( A^2_t - a^2A^2_\theta \big), 
\\
\Bigg( {A_t \over t \, a} \Bigg)_t - \Bigg( {a \, A_\theta  \over t } \Bigg)_\theta 
& = -\frac{4}{ta} \, \big( U_tA_t - a^2U_\theta A_\theta \big),
\endaligned
\ee
together with an ordinary differential equation for the coefficient $a=a(t, \theta)> 0$: 
\bel{constraint-II}
a_t = - t (1 - k^2) \, a \, \rho. 
\ee

Furthermore, the full model of interest contains an additional unknown $\nu=\nu(t,\theta) \in \RR$ whose evolution is
\textsl{decoupled} from the above system, i.e. 
\bel{evolution-II}
\Bigg( t \, {\nu_t \over a} + t^2 \, (1-k^2) {\rho \over a} \Bigg)_t 
    - \big( t \, a \nu_\theta \big)_\theta 
= 2at \, U_\theta^2 + {1 \over 2at} e^{4U} A_t^2 + t \, {\rho \over a} \, \frac{1-k^2+2k^2v^2}{1 - v^2}, 
\ee
supplemented with the constraints 
\bel{constraint-III}
\aligned 
& \nu_t  = U_t + {\Sigma_0 \over t}  
+ t\rho \, \frac{k^2 + v^2}{1 - v^2},
\\
& \nu_\theta  = U_\theta + {\Sigma_1 \over t}  
                             - t \, {\rho \over a} \, \frac{(1+k^2)v}{1 - v^2}. 
\endaligned
\ee

This completes the description of the \textsl{Euler-Gowdy system}, as we call it here. 
We refer to \eqref{fluid-II}-\eqref{evolution-I}-\eqref{constraint-II} as the ``essential Euler-Gowdy equations'', which suggest us to solve first for the fluid unknowns $\rho, v$ and the coefficients $U, A, a$. Next, we find the unknown $\nu$ from the remaining equation \eqref{evolution-II}, while \eqref{constraint-III} can be treated as constraints 
that can be ``propagated'' and need to be checked on the initial data only.

We emphasize that the right-hand sides of the equations \eqref{fluid-II} generalize the standard Euler equations for relativistic fluids, a system that has been studied extensively in, for instance, Smoller and Temple \cite{SmollerTemple-1993} and Makino and Ukai \cite{MakinoUkai-1995}. Namely, this is clear if, in \eqref{fluid-II}, we formally replace both the coefficient $a$ 
and the variable coefficient $t$ by $1$.


\subsection{Two energy functionals} 

The system under consideration admits two functionals
\be
\aligned
E_1(t) &\coloneqq \int_{T^1}  e \, d\theta
\qquad \text{and} \qquad
E_2(t) \coloneqq \int_{T^1}  \big( e + T \big) \, d\theta
\endaligned
\ee
with
$$
\aligned
e & \coloneqq \frac{1}{at^2}\bigg(\Sigma_0+\frac{1}{4}\bigg)
= {1 \over a} \, \Bigg( \big(U_t - \frac{1}{2t}\big)^2 + a^2U_\theta^2  + \frac{e^{4U}}{4t^2}\big(A_t^2 + a^2 A_\theta^2\big)\Bigg)
\geq 0,
\\
T & \coloneqq {\rho \over a} \, \frac{1+k^2v^2}{1 - v^2}
\geq 0.
\endaligned
$$
These density functions $e$ and $e + T$ obey \textsl{balance laws of energy}.
Those associated with~$E_1$ are
\bel{eq:use0}
\aligned
e_t + \big( af\big)_\theta 
& =  - {2 \over t} \, \et + \frac{a_t}{a}e, 
\\
f_t + \big( a e \big)_\theta 
& = - \frac{f}{t} + \frac{a_t}{a} f,
\endaligned
\ee
in terms of
\bel{eq:deffandet}
\aligned
f & \coloneqq -\frac{1}{t^2}\Sigma_1 = - 2 \, \big(U_t - \frac{1}{2t}\big)U_\theta  - \frac{e^{4U}}{2t^2}A_tA_\theta,
\\
\et & \coloneqq {1 \over a} \, \Bigg( \big(U_t - \frac{1}{2t}\big)^2  + \frac{e^{4U}}{4t^2} A_t^2 \Bigg).
\endaligned
\ee
In \eqref{eq:use0}, we can eliminate $a_t$ and write 
\bel{eq:use0-rho}
\aligned
e_t + \big( af\big)_\theta 
& =  - {1 \over t} \Big( 2 \, \et + (1-k^2) \, t^2 \rho \, e \Big), 
\\
f_t + \big( a e \big)_\theta 
& = - \frac{f}{t} \Big( 1 + (1-k^2) \, t^2 \rho \Big). 
\endaligned
\ee

Next, introducing
$$
M \coloneqq \frac{\rho}{a} \, {(1+k^2)v \over 1 - v^2},
\qquad
S \coloneqq \frac{\rho}{a} \, \frac{k^2 + v^2}{1 - v^2} = T + \frac{a_t}{t\,a^2},
$$
we also have similar equations associated with $E_2$: 
\bel{eq:use}
\aligned
(e + T)_t + \big( a (f + M) \big)_\theta 
& = - \frac{1}{t}\biggl( 2 \, \et + T + (3k^2 + 1){\rho \over 4a}\biggr),
\\
(f + M)_t + \big( a (e + S) \big)_\theta 
& = - \frac{1}{t}(f + M).
\endaligned
\ee

Both functionals turn out to enjoy monotonicity properties (i.e.~to be non-increasing or non-decreasing depending on the sign of $t$), as is clear from
\bel{der1-II}
\aligned
\frac{d E_1}{dt} (t) 
& =  - {1 \over t} \int_{T^1} \Big( 2 \, \et + (1-k^2) \, t^2 \rho \, e \Big) \, d\theta,
\\
 \frac{d E_2}{dt} (t) 
& = -\frac{1}{t}\int_{T^1} \Bigg( 2 \, \et+ T + (3k^2 + 1) {\rho \over 4 a} 
\Bigg)
\, d\theta.
\endaligned
\ee
Similar statements can be written by introducing a time-dependent weight. For instance,
\bel{der2-II-weight}
\aligned
 \frac{d}{dt} \Big( t^2 E_2(t) \Big)
= t \int_{T^1}\Bigg(
2a\,\bigg(U_\theta^2+\frac{e^{4U}}{4t^2}A_\theta^2\bigg)
+{\rho \over 4 a (1-v^2)} \Big( 3(1-k^2) + (7k^2+1) v^2 \Big)\Bigg)
\, d\theta.
\endaligned
\ee
 
For weak solutions, all of the above equations should be understood in the sense of distributions. We are interested in the initial value problem with data prescribed at some constant time $t_0 \neq 0$ and it is natural to distinguish the cases $t\gtrless 0$. 
Later in Section~\ref{sec:42fluid} we will see that the coefficient $a> 0$ may blow up in finite time, so that
(weak) solutions may exist only on a bounded interval of time. By imposing that the initial data have finite energy, we now arrive at the following result.

\begin{proposition}[Energy functionals for the Euler-Gowdy system]
Consider weak solutions to the Euler-Gowdy system and denote by $\tmax$ the maximal time of existence within which the function $a$ remains positive and locally bounded. 

\bei 

\item For solutions defined within a time domain $[t_0, \tmax) \subset [t_0, +\infty)$ corresponding to the initial value problem with initial data prescribed at a time $t_0> 0$, both energy functionals $E_1$ and $E_2$ are non-increasing
within the time interval $[t_0, \tmax)$. 

\item For solutions defined within a time domain $t \in [t_0, \tmax) \subset [t_0, 0)$ corresponding to the initial value problem with initial data prescribed at a time $t_0<0$, the weighted energy functional $t^2 E_2$ is non-increasing
within the time interval $[t_0, \tmax)$ while $E_1$ is non-decreasing but bounded above by $E_2$. 

\eei
\noindent Consequently, the following norm of the solutions is controlled within their interval of definition: 
\be
\aligned
\int_{T^1} \Bigg( 
{\rho \over 1 - v^2}
+
U_t^2 + a^2 U_\theta^2  + e^{4U} \big(A_t^2 + a^2 A_\theta^2\big)\Bigg) 
\, {d \theta \over a} 
\endaligned
\ee
and, in particular, the total mass density of the fluid is uniformly controlled within $[t_0, \tmax)$. 
\end{proposition}

\begin{remark} 
Concerning \eqref{eq:use}, we observe that 
$$
- (3k^2 + 1){\rho \over 4at}
= {1 + 3k^2 \over 1 - k^2} {a_t \over 4 t^2 a^2} 
= {1 + 3k^2 \over 1 - k^2} \Bigg( \bigg(- {1 \over 4t^2a} \bigg)_t - {1 \over 2 t^3 a} \Bigg), 
$$
so that the first equation in \eqref{eq:use} admits the alternative form
\be
\Bigg(
e + T + {1 + 3k^2 \over 1 - k^2} {1 \over 4t^2a} \Bigg)_t
+ \big( a (f + M) \big)_\theta 
 = - \frac{1}{t} \Bigg(
2 \, \et + T + {1 + 3k^2 \over 1 - k^2} {1 \over 2 t^2 a} 
\Bigg). 
\ee
This identity allows us to also control the $L^1$ norm of the function $1/a$. 
\end{remark}


\subsection{The structure of the Gowdy equations} 
\label{sec:GowdyEqua}

We consider first the Gowdy equations \eqref{evolution-I} and \eqref{constraint-II}, which we find it convenient to rewrite in terms of the new unknowns
\be
P \coloneqq 2 U - \ln t, \qquad Q \coloneqq A. 
\ee
We obtain the two coupled wave equations
\bel{33-eq2} 
\aligned
\Bigg( t \, {P_t \over a}\Bigg)_t  - \big( t a P_\theta \big)_\theta 
& = {t\,e^{2P} \over a} \, \big( Q^2_t - a^2 Q^2_\theta \big), 
\\
\Bigg( t \, {Q_t \over a} \Bigg)_t - \big( t \, a \, Q_\theta \big)_\theta 
& = 
 -\frac{2t}{a} \, \big( P_t Q_t - a^2 P_\theta Q_\theta \big),
\endaligned
\ee
in which the function $a$ is determined from the fluid density $\rho$ by the ordinary differential equation
\bel{constraint-II-PQ}
{a_t \over a}= - t (1 - k^2) \, \rho. 
\ee
Denoting by $a_0$ the value of $a$ at some initial time $t_0 \neq 0$, we can write 
\be
a(t,\theta) = a_0(\theta) \, e^{- (1-k^2) \int_{t_0}^t s\, \rho(s, \theta) \, ds}, \qquad 
t \in [t_0, \tmax), \, \theta \in T^1.
\ee  
This is a system of two coupled nonlinear wave equations, which can at least be solved locally in time once sufficiently smooth initial conditions are prescribed at $t_0$.  Observe that the right-hand sides of \eqref{33-eq2} are 
null terms\footnote{We call here ``null terms'' the quadratic expressions $(P_t^2 - a^2 P_\theta^2)$, 
$(Q_t^2 - a^2 Q_\theta^2)$, and $(P_t Q_t - a^2 P_\theta Q_\theta)$, where the factors $a^2$ are due to the non-trivial metric. Such terms are known to have good regularity and decay properties.}. 

\bei 

\item Once the metric coefficients $P, Q$ are known, the third metric coefficient  can be determined. It is convenient to define 
\be
\lambda \coloneqq 4 (\nu - U) + \ln t, 
\ee
and from \eqref{constraint-III} we see that $\lambda$ is obtained from 
\bel{33-weakconstraints-rho}
\aligned
{1 \over t} \lambda_t 
& = P_t^2 + a^2 P_\theta^2  
            + e^{2P} \big( Q_t^2 + a^2 Q_\theta^2 \big) + 4 \rho \, {k^2 + v^2 \over 1-v^2}, 
\\
{1 \over t} 
\lambda_\theta 
& = 2 P_t P_\theta + 2 e^{2P} Q_t Q_\theta - {4 \rho \over a} \, \frac{(1+k^2)v}{1 - v^2}. 
\endaligned
\ee  

\item It follows also that 
$$
\aligned 
& \Bigg(
2 t P_t P_\theta + 2 t e^{2P} Q_t Q_\theta  
\Bigg)_t 
- \Bigg(
t P_t^2 + t a^2 P_\theta^2  
            + t e^{2P} \big( Q_t^2 + a^2 Q_\theta^2 \big) 
\Bigg)_\theta 
\\
& = 
\Bigg( 4 {\rho t \over a} \, \frac{(1+k^2)v}{1 - v^2} \Bigg)_t 
+ \Bigg( 4 \rho t \, {k^2 + v^2 \over 1-v^2}  \Bigg)_\theta 
= - 4 \rho \, (1-k^2) \, \Sigma_1
\\
& = - 2 \rho \, (1-k^2)  \Bigg( t^2 \, P_t P_\theta + t^2 e^{2P} \, Q_t Q_\theta\Bigg)
=  {a_t \over a} \Bigg( 2t P_t P_\theta + 2t e^{2P} \, Q_t Q_\theta
\Bigg), 
\endaligned
$$
and therefore 
\be
\aligned 
& a
\Bigg(
{2 t \over a} \Big(  P_t P_\theta + e^{2P} Q_t Q_\theta  \Big) 
\Bigg)_t 
- \Bigg(
t \Big( 
P_t^2 + a^2 P_\theta^2  
            + e^{2P} \big( Q_t^2 + a^2 Q_\theta^2 \big) \Big)
\Bigg)_\theta 
 = 0. 
\endaligned
\ee
Remarkably, this equation does not depend explicitly on the velocity of the fluid, but only on its mass density $\rho$ via the function $a$. This equation is also equivalent to the second equation in \eqref{eq:use0}. 

\item In addition, we have a nonlinear wave equation for $\lambda$, which can be derived from \eqref{evolution-II} and is equivalent to the balance law for the energy: 
\bel{eq:use0-r}
\aligned
a \, \Bigg(
 {1 \over a^2} \, \Bigg(
P_t^2+ a^2 P_\theta^2 + e^{2P} \big(Q_t^2 + a^2 Q_\theta^2 \big) 
\Bigg)
\Bigg)_t 
- \big( 2a P_t P_\theta  + 2 a e^{2P} Q_t Q_\theta \big)_\theta 
& =  - {2 \over a t} \,  \Big( P_t^2  + e^{2P} Q_t^2 \Big).
\endaligned
\ee

\eei 
 
For sufficiently smooth solutions, the equation \eqref{eq:use0-r} is a consequence of \eqref{33-eq2}--\eqref{33-weakconstraints-rho}, 
yet it will play a role in controlling the $H^1$ norm of $P,Q$.   
For the second equation in \eqref{33-weakconstraints-rho} to lead to a function $\lambda$ defined on the torus $T^1$, 
the following condition is required: 
\bel{33-integralnulle}
\int_{T^1}  \Bigg( P_t P_\theta + e^{2P} Q_t Q_\theta 
- {2 \rho \over a} \, \frac{(1+k^2)v}{1 - v^2}
\Bigg)(t_0, \theta) \, d\theta = 0. 
\ee  
However, if it holds at some ``initial'' time then it remains true for all times, 
as follows by integrating the second equation in \eqref{eq:use}, i.e.
\be
{d \over dt} \int_{T^1}  t \Bigg( P_t P_\theta + e^{2P} Q_t Q_\theta 
- {2 \rho \over a} \, \frac{(1+k^2)v}{1 - v^2}
\Bigg) \, d\theta =0.
\ee 
This is true for all smooth solutions to \eqref{33-eq2}, as well as for all weak solutions (in the sense below).
Throughout, the condition \eqref{33-integralnulle} is always assumed to hold.


\section{The Gowdy equations} 
\label{sec:3}

\subsection{The notion of \HoneWoneone\ weak solutions} 
\label{sec:H1W1}

For clarity in the presentation, we study first the case where $a$ is a constant which we normalize to be $a \equiv 1$.  
To define a concept of weak solutions, 
we investigate the following \textsl{Gowdy equations} with unknowns $(P,Q,\lambda)$ 
\bel{33-eq2-FULL}
\aligned 
P_{tt} - P_{\theta\theta} + {1 \over t} \, P_t  
& =  e^{2P} (Q_t^2 - Q_\theta^2), 
\\ 
Q_{tt} - Q_{\theta\theta} + {1 \over t} \, Q_t  
& = - 2 (P_t Q_t - P_\theta Q_\theta), 
\endaligned
\ee
and 
\bel{33-eq2-FULL-2}
\aligned
\lambda_{tt} - \lambda_{\theta\theta}
& = - P_t^2 + P_\theta^2  
           + e^{2P} \big( - Q_t^2 + Q_\theta^2 \big),
\endaligned
\ee
coupled to the two constraint equations  
\bel{33-weakconstraints}
\aligned
{1 \over t} \lambda_t 
& = P_t^2 + P_\theta^2  
            + e^{2P} \big( Q_t^2 + Q_\theta^2 \big), 
\\
{1 \over t} 
\lambda_\theta 
& = 2 P_t P_\theta + 2 e^{2P} Q_t Q_\theta.  
\endaligned
\ee 
The evolution of the coefficient $\lambda$ is thus determined by solving the wave equations \eqref{33-eq2-FULL-2}, but $\lambda$ can also be recovered from the constraint equations \eqref{33-weakconstraints} (cf.~Remark~\ref{rem-33} below). We refer to the two equations \eqref{33-eq2-FULL} as the \textsl{essential Gowdy equations.} 

It is worth observing that the above system provides a solution to the Euler-Gowdy system when the mass density is chosen to vanish identically, while the velocity can be arbitrary.  In particular, in view of \eqref{constraint-II-PQ}, the function $a$ then is constant in time and by introducing a change of variable in the variable $\theta$ (only) one can reduce the problem to a constant function $a$. 

We use a standard notation for Lebesgue and Sobolev spaces, such as $H^1$ (functions with square-integrable derivatives), $W^{1,1}$ (functions with integrable derivatives), etc. First of all, we need to introduce a notion of initial data set.  The Gowdy equations can solved in both forward and backward time directions, so without loss of generality we solve from a positive initial time. 

\begin{definition} 
\label{def-weak-Gowdy-data}
Consider the Gowdy equations \eqref{33-eq2-FULL} and \eqref{33-eq2-FULL-2} and fix some  time $t_0> 0$.
A set of functions $(\barP, \barP_0, \barQ, \barQ_0)$ defined on $T^1$ 
is called a \emph{$H^1$-$W^{1,1}$ essential Gowdy initial data set} 
if   
\bel{33-intialregu}  
\barP, \barQ \in H^1(T^1),
\qquad 
\barP_0, \barQ_0 \in L^2(T^1). 
\ee 
Furthermore, a set of functions $(\barP, \barQ,  \barlambda, \barP_0, \barQ_0, \barlambda_0)$ defined on $T^1$ 
is called a \emph{$H^1$-$W^{1,1}$ Gowdy initial data set} if \eqref{33-intialregu}  holds together with
\bel{33-weakconstraints-calcul}
\aligned
{1 \over t_0} \barlambda_0
& = \barP_0^2 + \barP_{\theta}^2  
            + e^{2 \barP} \big( \barQ_0^2 + \barQ_{\theta}^2 \big), 
\\
{1 \over t_0} \barlambda_{\theta}
& = 2 \barP_0 \barP_{\theta} + 2 e^{2\barP} \barQ_0 \barQ_{\theta},
\endaligned
\ee
which implies $\barlambda_0 \in L^1(T^1)$ and $\barlambda \in W^{1,1}(T^1)$. 
\end{definition} 

By Sobolev's embedding theorem, all of our initial data and solutions are in $L_\loc^\infty$, so that the terms $e^{2 \barP}$ and $e^{2P}$ are locally bounded functions (and in fact H\"older continuous functions). 
We then introduce a notion of solution at the same level of weak regularity.

\begin{definition} 
\label{def-weak-Gowdy-solution}
Consider the Gowdy equations \eqref{33-eq2-FULL}. 
Given any interval of time $I \subset (0, + \infty)$, a pair of functions $(P,Q)$ defined on $I \times T^1$ and satisfying 
\bel{eq-weakPQ} 
\aligned
&  P,Q \in L_\loc^\infty(I, H^1(T^1)), 
\qquad 
P_t, Q_t \in L_\loc^\infty(I, L^2(T^1)), 
\endaligned
\ee
 is called a $H^1$ \emph{weak solution to the essential Gowdy equations} if the equations
 \eqref{33-eq2-FULL} hold in the sense of distributions\footnote{The left-hand sides of \eqref{33-eq2-FULL} are distributions in
$H^{-1}$, 
while the right-hand sides of \eqref{33-eq2-FULL} are integrable functions on $T^1$. (Recall that $P, Q$ are known to be bounded.)}. 
A triple of functions $(P,Q, \lambda)$ defined on $I \times T^1$ and satisfying \eqref{eq-weakPQ} and 
\bel{eq-weaklambda}
\aligned 
& \lambda \in L_\loc^\infty(I, W^{1,1}(T^1)), 
\qquad 
\lambda_t \in L_\loc^\infty(I, L^1(T^1)), 
\endaligned
\ee
is called a $H^1$-$W^{1,1}$ \emph{weak solution to the Gowdy equations}  
 if the evolution equations \eqref{33-eq2-FULL}-\eqref{33-eq2-FULL-2} hold in the sense of distributions while the constraint equations \eqref{33-weakconstraints} hold as equalities between $L^1$ functions. 
\end{definition}

Finally, we state a notion of solution to the initial value problem, when an initial condition is prescribed: 
\be
\aligned
& P(t_0,\cdot) = \barP, \quad && Q(t_0,\cdot) = \barQ, \quad && \lambda(t_0,\cdot) = \barlambda,
\\
& P_t(t_0,\cdot) = \barP_0, \quad && Q_t(t_0,\cdot) = \barQ_0,  \quad  && \lambda_t(t_0,\cdot) = \barlambda_0. 
\endaligned
\ee

\begin{definition} 
\label{def:initialdata}
Under conditions in Definitions \ref{def-weak-Gowdy-data} and \ref{def-weak-Gowdy-solution}, a weak solution $(P,Q)$ to the essential Gowdy equations (defined for $t \in I=[t_0, t_1)$)
is said to assume the prescribed initial data set $(\barP, \barQ, \barP_0, \barQ_0)$ at the time $t_0$ if\footnote{
From now on, $P(t_0)$ stands for $P(t_0,\cdot)$, etc.}
\bel{eq:3417}
\aligned
& \lim_{t \to t_0 \atop t >t_0} {1 \over t-t_0} \int_{t_0}^t 
\Big(
\| P(s) - \barP \|^2_{H^1(T^1)} + \| Q(s) - \barQ \|^2_{H^1(T^1)} \Big) \, ds = 0, 
\\
& \lim_{t \to t_0 \atop t >t_0} {1 \over t-t_0} \int_{t_0}^t 
\Big(
\| P_t(s) - \barP_0 \|^2_{L^2(T^1)} + \| Q_t(s) - \barQ_0 \|^2_{L^2(T^1)} \Big) \, ds = 0.  
\endaligned
\ee
Furthermore, a weak solution $(P,Q,\lambda)$ to the Gowdy equations (defined for $t \in I=[t_0, t_1)$)
is said to assume the prescribed initial data set  $(\barP, \barQ, \barP_0, \barQ_0, \barlambda, \barlambda_0)$ 
if \eqref{eq:3417} holds together with 
\be
\aligned
& \lim_{t \to t_0 \atop t >t_0} {1 \over t-t_0} \int_{t_0}^t 
\Big(
\| \lambda(s) - \barlambda \|_{W^{1,1}(T^1)} + \| \lambda_t(s) - \barlambda_0 \|_{L^1(T^1)}
\Big) \, ds = 0.
\endaligned
\ee
\end{definition}

We emphasize that Definition~\ref{def-weak-Gowdy-solution} only requires that the solution have $L^\infty$ regularity in time, which does not allow to define the trace at $t=t_0$ in a pointwise sense and this is why we have introduced an average in time in Definition~\ref{def:initialdata}. However, by using the wave equations\footnote{The desired regularity follows from the equation satisfied by the energy (cf.~Section~\ref{sec:Gowdy-Volu}).}
 satisfied by $P,Q,\lambda$, we can deduce the additional regularity 
\bel{eq:3130}
\aligned
& P,Q \in C^0(I, H^1(T^1)), \qquad 
P_t, Q_t \in C^0(I, L^2(T^1)),  \qquad 
\\
& \lambda \in C^0(I, W^{1,1}(T^1)),  \qquad 
\lambda_t \in C^0(I, L^1(T^1)), 
\endaligned
\ee
so that the integrands in \eqref{eq:3417} are in fact continuous in $t$,
and the initial conditions hold in the stronger sense: 
\bel{eq:66604}
\aligned
& \lim_{t \to t_0 \atop t >t_0} 
\Big( \| P(t) - \barP \|^2_{H^1(T^1)} + \| Q(t) - \barQ \|^2_{H^1(T^1)} \Big) = 0, 
\\
& \lim_{t \to t_0 \atop t >t_0}  
\Big( \| P_t(t) - \barP_0 \|^2_{L^2(T^1)} + \| Q_t(t) - \barQ_0 \|^2_{L^2(T^1)} \Big) = 0.  
\endaligned
\ee
From \eqref{eq:66604} it is then immediate that 
\bel{eq:31333}
\aligned
& \lim_{t \to t_0 \atop t >t_0}  
\Big(
\| \lambda(t) - \barlambda \|_{W^{1,1}(T^1)} + \| \lambda_t(t) - \barlambda_0 \|_{L^1(T^1)}
\Big) = 0.
\endaligned
\ee 

\begin{remark}
\label{rem-33} 
1. The same function $\lambda$ is obtained by using either the evolution equation for $\lambda$ or
 the constraints \eqref{33-weakconstraints}. Up to an integration constant, we have 
\be
\aligned
\lambda(t,\theta) = {}
  & \int_{\theta_0}^\theta 2 t_0 \, \big( P_t P_\theta + e^{2P} Q_t Q_\theta\big)(t_0,\theta') \, d\theta'
\\& 
+ \int_{t_0}^t t' \ \big(P_t^2 + P_\theta^2 + e^{2P} \big( Q_t^2 + Q_\theta^2 \big)\big)(t',\theta)\, dt', 
\endaligned
\ee
which, in agreement with Definition~\ref{def-weak-Gowdy-solution}, gives us a function $\lambda \in L_\loc^\infty(I, W^{1,1}(T^1))$.
The existence of the time integral will be established in Step 2 of the proof of Theorem~\ref{33-theoG}.

2. Everywhere in this text, $L^2$, $H^1$, \ldots regularity is always understood as  $L^2$, $H^1$, \ldots regularity with respect to the \textsl{space variable,} while the time regularity is easily inferred from this regularity and often will not be explicitly indicated.
\end{remark}


\subsection{First-order formulation of the essential Gowdy equations}
\label{sec:Gowdy-conserv}

By definition, the essential Gowdy equations are the two equations \eqref{33-eq2-FULL}.
Sufficiently smooth solutions of this subsystem determine the full dynamics of the solutions, so it is natural to investigate this subsystem first. For any solution, the following identities hold
\bel{key-observe}
\big( t f(P,Q)P_t + t g(P,Q)Q_t \big)_t
- \big( t f(P,Q)P_\theta + t g(P,Q)Q_\theta \big)_\theta = 0
\ee
in which the possible choices of the functions $f, g$ are parametrized by a quadratic function $F(Q)$ and are
\be
f(P,Q) \coloneqq F'(Q) , \qquad
g(P,Q) \coloneqq \frac{1}{2} F''(Q) - F(Q) e^{2P} .
\ee
In particular, for $F(Q)=Q$ and $F(Q)=-1$ we can write
\bel{eq:33-8} 
\aligned
\Big( t \big( P_t - e^{2P} Q Q_t \big) \Big)_t - \Big( t \big( P_\theta - e^{2P} Q Q_\theta \big) \Big) _\theta = 0, 
\\
\big( t e^{2P}  Q_t \big)_t - \big( t e^{2P} Q_\theta \big)_\theta  = 0. 
\endaligned
\ee
These equations are equivalent to \eqref{33-eq2-FULL}.
We now rewrite the equations in terms of $Q$ and $R\coloneqq e^{-2P}+Q^2$.

\begin{definition} 
\label{def-conser-Gowdy}
The \emph{first-order formulation} of the Gowdy equations, by definition,  
consists of the following two second-order conservation laws in $(R,Q)$ 
\bel{eq:33-8-deux} 
\aligned
\big( t \, \Omega(R,Q)^{-2}   R_t \big)_t - \big( t \, \Omega(R,Q)^{-2}   R_\theta \big) _\theta = 0, 
\\
\big( t \, \Omega(R,Q)^{-2} Q_t \big)_t - \big( t \, \Omega(R,Q)^{-2} Q_\theta \big)_\theta  = 0, 
\endaligned
\ee  
with a coefficient
\be
\Omega^2 \coloneqq e^{-2P} = R - Q^2.  
\ee
\end{definition}

We set  
\be
\aligned 
& \Phi \coloneqq (R_0, R_1, Q_0, Q_1) = (R_t, R_\theta, Q_t, Q_\theta), 
\endaligned
\ee
which we refer to as the \textsl{first-order variables}.
The ``non-local'' coefficient $\Omega$ is determined by integration of $\Phi$ from some initial time $t_0$
(with $\theta \in (0,1)$): 
\bel{eq:33-9b} 
\aligned
R(t, \theta) & = \int_{t_0}^t R_0(t', 0) \, dt' + \int_0^\theta R_1(t, \theta') \, d\theta',
\\
Q(t, \theta)  & = \int_{t_0}^t Q_0(t', 0) \, dt' + \int_0^\theta Q_1(t, \theta') \, d\theta',
\endaligned
\ee
yielding us an expression for $\Omega^2 = R-Q^2$. 
We have arrived at the following.

\begin{corollary} The essential Gowdy equations take the form of 
a first-order hyperbolic system of four conservation laws 
\bel{eq:33-9} 
\aligned
\Big( t \, \Omega^{-2} R_0 \Big)_t - \Big(  t \, \Omega^{-2} R_1 \Big)_\theta = 0,
\\
\big( R_1 \big)_t - \big(  R_0 \big)_\theta  = 0, 
\\
\Big(  t \, \Omega^{-2} Q_0 \Big)_t - \Big(  t \, \Omega^{-2} Q_1 \Big)_\theta = 0,
\\
\big( Q_1 \big)_t - \big( Q_0 \big)_\theta  = 0, 
\\
\endaligned
\ee 
in which the ``non-local'' coefficient $\Omega$ is given by the integral formulas \eqref{eq:33-9b}. 
\end{corollary}

Our first-order formulation determines the full evolution of the Gowdy solutions. 
In contrast to it, all previous studies in the literature on Gowdy symmetry were based on the non-conservative equations \eqref{33-eq2-FULL}. 
Some remarks are in order: 
\bei 
\item The second and fourth equations in \eqref{eq:33-9} are trivial compatibility relations, but are necessary to ``close'' the system. For consistency, it is also required that 
\be
\int_{T^1}  R_1(t, \theta') \, d\theta' = \int_{T^1}  Q_1(t, \theta') \, d\theta' = 0, 
\ee
which do hold for all times provided they hold at some initial time $t_0$. 

\item In agreement with Definition~\ref{def-weak-Gowdy-solution}, we consider weak solutions to \eqref{eq:33-9}--\eqref{eq:33-9b}, satisfying the $L^2$ regularity conditions 
\be
\Phi \in L^\infty(I, L^2(T^1)). 
\ee
The coefficient $\Omega \in L^\infty(I, H^1(T^1))$ is then more regular, and we may expect to handle \eqref{eq:33-9} by techniques of analysis for linear hyperbolic systems with non-constant coefficient, although the coefficient $\Omega^2$ is actually determined from an integral expression in the main unknowns. A further reduction will be presented later in Section~\ref{sec:6}.
\eei 

\begin{remark}  
The choice of conservative variables is not unique as is clear from \eqref{key-observe}.
 If the Gowdy equations are interpreted as a wave map system posed on the hyperbolic space, the conservation laws can also be related to the symmetries of the hyperbolic space via Noether's theorem.
The compatibility relations are also only one choice among the general identities (valid for arbitrary  $G=G(P,Q)$): 
\be
\aligned
& \big( G_P'(P,Q) P_\theta + G_Q'(P,Q) Q_\theta \big)_t 
- \big( G_P'(P,Q) P_\theta + G_Q'(P,Q) Q_\theta \big)_t 
\\
& = G(P,Q)_{\theta t} - G(P,Q)_{t\theta}
= 0.
\endaligned
\ee
\end{remark}


\subsection{The quadratic formulation of the Gowdy equations}
\label{sec:Gowdy-quadra}

Before we proceed further, we introduce an alternative formulation of the Gowdy equations, which also puts some light on their structure. Unlike our first-order formulation, this second formulation is not based on conservation laws but has the advantage of involving quadratic nonlinearities only. 
We introduce the new dependent variables
\be
(P_0, P_1, S_0, S_1) \coloneqq  (P_t, P_\theta, e^P Q_t, e^P Q_\theta),  
\ee
and we obtain the following structure of the essential equations
\bel{eq-quadra1}
\aligned
P_{0t} - P_{1\theta} + {P_0 \over t} & = (S_0)^2 - (S_1)^2, 
\\
P_{1t} - P_{0\theta} & = 0, 
\\
S_{0t} - S_{1\theta} + {S_0 \over t}& = - P_0 S_0 + P_1 S_1, 
\\
S_{1t} - S_{0\theta} & =  P_0 S_1 - P_1 S_0.  
\endaligned
\ee
On the other hand, for the equations satisfied by $\lambda$, we find 
\bel{eq-quadra2}
\aligned
\lambda_{0t} - \lambda_{1\theta} 
& =  - (P_0)^2 + (P_1)^2 - (S_0)^2 + (S_1)^2,
\\
\lambda_{1t} - \lambda_{0\theta} & = 0, 
\\
{\lambda_0 \over t} 
& =  (P_0)^2 + (P_1)^2 + (S_0)^2 + (S_1)^2, 
\\ 
{\lambda_1 \over t} 
& =   2 (P_0 P_1 + S_0 S_1), 
\endaligned
\ee
after setting
\bel{eq:lambdalambda} 
\lambda_0 = \lambda_t, \qquad \lambda_1 \coloneqq \lambda_\theta. 
\ee  

\begin{definition} 
\label{def-quadr}
The set of equations \eqref{eq-quadra1} and \eqref{eq-quadra2} with unknown $\Psi\coloneqq(P_0, P_1, S_0, S_1, \lambda_0, \lambda_1)$ is referred to as the \emph{quadratic formulation of Gowdy equations,} 
and $\Psi$ is referred to\footnote{With some abuse of terminology, we also refer to $(P_0, P_1, S_0, S_1)$ as the quadratic variable when we only consider the essential Gowdy equations.}
as the \emph{quadratic variable.} 
\end{definition} 

Clearly, seeking for $H^1$-$W^{1,1}$ weak solutions $P, Q, \lambda $ in the sense of Definition~\ref{def-weak-Gowdy-solution} is equivalent to seeking for $L^2$-$L^1$ \textsl{weak solutions} 
to \eqref{eq-quadra1} and \eqref{eq-quadra2} such that 
\be
P_0, P_1, S_0, S_1 \in L^\infty(I, L^2(T^1)), 
\qquad 
\lambda_0, \lambda_1 \in L^\infty(I, L^1(T^1)).
\ee 
Our formulation of the Gowdy equations involves only quadratic expressions and we are able here to eliminate the factor $e^{2P}$ that appears in the original formulation. In the following section, we will also give an energy functional that is fully quadratic in $\Psi$.


\subsection{A weighted energy functional for Gowdy solutions}
\label{sec:Gowdy-Volu}

\paragraph*{The energy functional.} 

 The weak regularity conditions we imposed are natural in order to give a meaning to the Einstein equations, but the same functional spaces arise with a
natural ``energy'' functional associated with this system. 
If $P, Q$ is an arbitrary weak solution to the essential Gowdy equations \eqref{eq:33-8}, then  it also satisfies the energy
identity
\bel{33-bal}
E(\Psi)_t + F(\Psi)_\theta = - \frac{2}{t} G(\Psi) \leq 0, 
\ee
in which the energy density, energy flux, and the energy dissipation are defined, respectively, by 
\be
\aligned
E(\Psi) 
&\coloneqq P_t^2 + P_\theta^2 + e^{2P} \big( Q_t^2 + Q_\theta^2  \big) = P_0^2 + P_1^2 + S_0^2 + S_1^2 \geq 0, 
\\
F(\Psi) &\coloneqq - 2 \big( P_t P_\theta + e^{2P} Q_t Q_\theta \big) = -2 (P_0P_1 + S_0S_1) \leq E(\Psi), 
\\
G(\Psi) &\coloneqq P_0^2 + S_0^2 \leq E(\Psi).  
\endaligned
\ee
It is convenient also to define 
\be
H(\Psi) \coloneqq P_1^2 + S_1^2  \leq E(\Psi)
\ee
and express the energy identity as 
\bel{33-bal-factor-t2}
\big(t^2 E(\Psi) \big)_t + \big(t^2 F(\Psi) \big)_\theta = 2t \, H(\Psi) \geq 0. 
\ee
The following observation is immediate, but provides us with an essential property of the Gowdy equations.

\begin{proposition}[Convexity of the energy of the Gowdy equations] 
The energy density 
$\Ecal(\Psi)$
is a (uniformly) convex function of the quadratic variable $\Psi= (P_0, P_1, S_0, S_1)$. 
\end{proposition}

The \textsl{Gowdy energy functional} at time $t \in I$, defined as
\be
\Ecal(t) \coloneqq \int_{T^1} E(\Psi)(t, \theta) \, d\theta = \| \Psi(t, \cdot)\|_{L^2(T^1)}^2,
\ee
is nothing but the $L^2$ norm of the quadratic variable $\Psi$. It also is equivalent to the $L^2$ norm of the first-order variable $\Phi$. We have 
\bel{eq-237}
\| \Phi(t, \cdot) \|_{L^2(T^1)} 
\lesssim 
\| \Psi(t, \cdot)\|_{L^2(T^1)}
\lesssim  
\| \Phi(t, \cdot) \|_{L^2(T^1)} 
\ee
with constants that depend only on the sup-norm of $(P,Q)$,
itself controlled by its norm in $H^1(T^1)\subset L^\infty(T^1)$. More precisely, we have the following: 
\bei 
\item In the \textsl{forward} time direction, the total energy is controlled by integrating  \eqref{33-bal} which yields the folowing  $L^2$ estimate: 
\bel{eq-238}
\aligned
& \| \Psi(t_2, \cdot) \|_{L^2(T^1)} \leq  \| \Psi(t_1, \cdot) \|_{L^2(T^1)}, \qquad t_1 \leq t_2.
\endaligned
\ee
\item In order to advance in the \textsl{backward} time direction, we must use the weighted version \eqref{33-bal-factor-t2}
and we obtain  
\bel{eq-239}
\aligned
& \| \Psi(t_1, \cdot) \|_{L^2(T^1)} \leq (t_2/t_1)^2  \| \Psi(t_2, \cdot) \|_{L^2(T^1)},  \qquad t_1 \leq t_2.
\endaligned
\ee 
\eei 
\noindent Hence, the $L^2$ norm of both the quadratic and the first-order variables of any solution to the Gowdy equations is bounded for compact time interval, once it is assumed to be bounded on any initial hypersurface of constant time.  The statement in the quadratic variables, of course, is particularly simple and our quadratic variables are particularly convenient in order to address global issues on the behavior of weak solutions. 

\paragraph*{The weighted energy functional.} We now observe that a variant of the basic energy functional can be introduced which have even stronger dissipation properties and therefore provide us with an additional integrability property for the solutions. 
We start by considering the functional
\be
\Vcal(t) \coloneqq \int_{(T^1)^3} \sqrt{ e^{\lambda/2} t^{3/2}} \, dx^2 dx^3 d\theta
= t^{3/4} \int_{T^1} e^{\lambda/4} \, d\theta. 
\ee 
In view of \eqref{33-weakconstraints}, we have the identity 
\be
{d \over dt} \Big( t^{-3/4} \Vcal(t) \Big) = {t \over 4} \int_{T^1} e^{\lambda/4} \, E(\Psi)  \, d\theta, 
\ee
which relates together the two functionals 
\be
\Vcalt(t) \coloneqq t^{-3/4} \Vcal(t), \qquad 
\Ecalt(t) \coloneqq  \int_{T^1} e^{\lambda/4} \,  E(\Psi)  \,  d\theta. 
\ee

\paragraph*{The integrability property.} 
Next, in order to compute the second-order derivative of the functional, from \eqref{33-bal} we deduce the weighted energy identity
$$
\aligned
\big( e^{\lambda/4} E(\Psi) \big)_t + \big( e^{\lambda/4} F(\Psi) \big)_\theta 
& = - \frac{2}{t} e^{\lambda/4} G(\Psi)
+ {1 \over 4} e^{\lambda/4} \Big( \lambda_t E(\Psi) + \lambda_\theta F(\Psi) \Big)
\endaligned
$$
so, in other words,  
\bel{eq-9938}
\aligned
& \big( e^{\lambda/4} E(\Psi) \big)_t + \big( e^{\lambda/4} F(\Psi) \big)_\theta 
  = - \frac{2}{t} e^{\lambda/4} G(\Psi) + {t \over 4} e^{\lambda/4}  N(\Psi)^2,
\endaligned
\ee 
in which we observe that the term $N(\Psi)^2$ (defined now) has indeed a \textsl{non-negative sign}, since 
\bel{eq:nullformnotation}
\aligned
N(\Psi)^2 
\coloneqq & \, E(\Psi)^2 - F(\Psi)^2 
\\
= \, 
& \big( P_t^2 - P_\theta^2 \big)^2+ e^{2P} \big( Q_t^2 - Q_\theta^2 \big)^2
+ \big( P_t - P_\theta \big)^2 e^{2P} \big( Q_t + Q_\theta \big)^2 
+ \big( P_t + P_\theta \big)^2 e^{2P} \big( Q_t - Q_\theta \big)^2
\\
= \, 
& \big( P_0^2 - P_1^2 \big)^2+ \big( S_0^2 - S_1^2 \big)^2
+ \big( P_0 - P_1 \big)^2 \big( S_0 + S_1 \big)^2 
+ \big( P_0+ P_1 \big)^2 \big( S_0 - S_1 \big)^2 \geq 0. 
\endaligned
\ee
Hence, we find that the evolution of the rescaled energy is given by 
\be
\aligned
{d \over dt} \Ecalt(t) & = - {2 \over t} \Gcal(t) + {t \over 4} \Ncal(t),
\\
\Gcal(t) & \coloneqq  \int_{T^1} G(\Psi) \,  e^{\lambda/4} \,  d\theta, 
\qquad  
\Ncal(t) \coloneqq  \int_{T^1} N(\Psi)^2 \,  e^{\lambda/4} \,  d\theta, 
\endaligned
\ee
which leads us to 
\be
\Ecalt(t) =  \Ecalt(t_0) - \int_{t_0}^t {2 \over s} \Gcal(s) ds + \int_{t_0}^t {s \over 4} \Ncal(s) ds,
\ee
in which \textsl{all but the term} involving $\Ncal$ are already controlled from the fundamental energy functional; cf.~\eqref{eq-238}--\eqref{eq-239}.  We have thus reached the following conclusion.

\begin{theorem}[Integrability property of Gowdy solutions] 
\label{theo:higherinte}
Weak solutions to the Gowdy equations 
 satisfies the additional \emph{space-time} integrability property
\bel{eq-high-int}
P_t^2 - P_\theta^2, 
\quad 
Q_t^2 - Q_\theta^2, 
\quad
\big( P_t - P_\theta \big) \big( Q_t + Q_\theta \big),
\quad
\big( P_t + P_\theta \big) \big( Q_t - Q_\theta \big)
 \in L^2_\loc (I \times T^1), 
\ee
which also implies 
\be
P_t Q_t - P_\theta Q_\theta,
\quad 
P_t Q_\theta - P_\theta Q_t
 \in L^2_\loc (I \times T^1). 
\ee
\end{theorem} 
 

\paragraph*{Yet another functional.}

The weighted measure $e^{\lambda/4} \,  d\theta$ we use here is natural, since it arises geometrically and allows us to interpret our functionals as geometric objects.
However, any strictly convex weight $\lambda\mapsto\omega(\lambda)\geq 0$ with bounded first-order derivative could also be used instead of~$e^{\lambda/4}$.
We introduce
\be
\Vcalt_{\omega}(t) \coloneqq \int_{T^1} \omega(\lambda) \, d\theta \geq 0. 
\ee
We find
\be
{d \over dt} \Vcalt_\omega(t) = t \int_{T^1} \omega'(\lambda) \, E(\Psi)  \, d\theta \eqqcolon  t  \Ecalt_\omega(t),
\ee 
while \eqref{eq-9938} generalizes to any $\omega$ and any real $\alpha$:
\bel{eq-9938-ab}
\aligned
& \big( t^{\alpha +1} \omega'(\lambda) E(\Psi) \big)_t
+ \big( t^{\alpha +1} \omega'(\lambda) F(\Psi) \big)_\theta
\\
& \quad =  
t^{\alpha} \omega'(\lambda) \, \Big( (\alpha-1) G(\Psi) + (\alpha +1) H(\Psi) \Big) 
    + t^{\alpha +2} \omega''(\lambda) \, N(\Psi). 
\endaligned
\ee 
Therefore, we obtain 
\be
\aligned
{d \over dt} \Bigg( t^\alpha {d \over dt} \Vcalt_\omega(t) \Bigg)
& =
{d \over dt} \Bigg(  t^{\alpha +1} \int_{T^1} \omega'(\lambda) \, E(\Psi)  \, d\theta
\Bigg) 
\\
& =  \int_{T^1} \Bigg(
t^{\alpha} \omega'(\lambda) \, \Big( (\alpha-1) G(\Psi) + (\alpha +1) H(\Psi) \Big) 
    + t^{\alpha +2} \omega''(\lambda) \, N(\Psi)
\Bigg)
  \, d\theta
\Bigg) 
\endaligned
\ee 
We summarize our conclusion as follows. 

\begin{proposition}[A convexity property] 
The functional $\Vcalt_{\omega}(t)$ satisfies the convexity property  
\be
{d \over dt} \Bigg( t^\alpha {d \over dt} \Vcalt_\omega(t) \Bigg) \geq 0 
\quad \text{ provided either } 
\begin{cases}
\omega' \geq 0 \text{ and } \alpha \geq 1, 
\\ 
\omega' \leq 0 \text{ and } \alpha \leq -1.  
\end{cases}
\ee 
\end{proposition}

\begin{remark}
While the integrability property was established in \cite{LeFlochSmulevici-2016}, the argument therein was different: it was not geometric-in-nature and was based on a splitting the energy into left-moving and right-moving parts. 
\end{remark}


\section{The Euler equations in Gowdy symmetry}
\label{sec:4}

\subsection{Spatially homogeneous solutions}
\label{homog_solutions}

In contrast with solutions to the Gowdy equations (with $a$ constant), solutions to the Euler-Gowdy system need not be defined for all values of the time variable $t$. We explain this fact by considering first the special class of spatially homogeneous solutions whose velocity component vanishes identically ($v=0$).
They are characterized by the equations 
\bel{ode_h} 
\aligned
\big(a^{-1}t\rho\big)_t 
& = a^{-1} t\rho(1-k^2) \, \bigg(-\frac{3 k^2 +1}{4t (1-k^2)} + atE_1(t) \bigg),
\\
\big(a^{-1}t(U_t - 1/(2t))\big)_t 
& =  \frac{e^{4U}}{2at}A_t^2, 
\\
\big(a^{-1}t^{-1}e^{2U}A_t \big)_t 
& = - 2\frac{e^{2U}}{a\, t} U_t A_t,  
\\
a_t & = -at\rho (1-k^2),
\endaligned
\ee 
where, in agreement with our earlier notation, 
\be
E_1 =  {1 \over a} \, \Bigg( \big(U_t - \frac{1}{2t}\big)^2 + \frac{e^{4U}}{4t^2} A_t^2 \Bigg). 
\ee
This first-order differential system in $\rho, a, U_t, A_t$ admits  a conserved energy, i.e. 
\bel{energy1}
\frac{d}{dt} \Big( a^{-1}t^2E_1(t) \Big) = 0, 
\ee
which turns out to play an essential role in the global behavior of this dynamical system.  

Namely, let us fix some time $t_0<0$. A direct analysis of \eqref{ode_h}  leads us to the following conclusion: 
when $E_1(t_0) > 0$ the solution exists on the whole interval $[t_0, 0)$, while 
when $E_1(t_0) = 0$ the evolution depends upon the value of the initial fluid density.
Our analysis is conveniently performed in terms of the normalized density 
\be
m \coloneqq \frac{4}{3}t^2\rho
\ee
and with respect to the rescaled time 
\bel{eq=newtime}
\tau = - \frac{3}{4}(1-k^2) \ln \big(\frac{t}{t_0}\big), \qquad \tau\in[0, + \infty). 
\ee 
In terms of the unknowns $(a, m)$ regarded  as  unknown functions of the variable $\tau$, the last two equations in \eqref{ode_h} take the form  
\bel{sys1}
\aligned
\frac{d(a^{-1})}{d\tau} & = -a^{-1}\,m,
\\ 
\frac{d(a^{-1}m)}{d\tau} & = -C_0a\,m   - a^{-1} \, m,
\endaligned
\ee
in which, in view of \eqref{energy1},  $C_0 \coloneqq \frac{4}{3}t_0^2E_1(t_0)a_0^{-1}$.  

We find that spatially homogeneous solutions to the Euler-Gowdy system are classified as follows: 
\begin{description}

\item[(A)] If $E_1(t_0) \neq 0$, then the function $a$ and, consequently, $tU_t, A_t, m$ remain globally bounded up to the time $t=0$ and, moreover, $m \rightarrow 0$ as $t \rightarrow 0$.

\item[(B)] If $E_1(t_0) = 0$, we have
$$
U_t=\frac{1}{2t},  \qquad A_t=0,  \qquad  \frac{m}{a} \big(\frac{t_0}{t}\big)^{3(1-k^2)/4}  =  \frac{m_0}{a_0}, \qquad \frac{1-m}{a} = \frac{1-m_0}{a_0}.
$$ 
Depending upon the initial value $m_0 = m(t_0)$, three cases may arise: 

\begin{itemize} 

\item[(i)] When $m_0<1$, 
the function $a$ remains bounded while $m \rightarrow 0$ as $t \rightarrow 0$. 

\item[(ii)] When $m_0=1$, the rescaled density 
 $m = 1$ is constant while $a = a_0\big(\frac{t_0}{t}\big)^{3(1-k^2)/4}$ blows up.

\item[(iii)] When $m_0>1$, both the function $a$ and the rescaled density $m$ blow up at some \textsl{non-vanishing critical time} $\tmax \in (t_0, 0)$. 
\end{itemize}

\end{description}

\noindent This shows that solutions to the Euler-Gowdy system may be defined on a bounded interval of (areal) time.
We only need an argument for the case (A): we observe that a linear combination of the two equations in \eqref{sys1} gives us 
\be
{d \over dt} \Big( C_0 a + a^{-1} m - a^{-1} \Big) = 0, 
\ee
so that studying \eqref{sys1} is actually equivalent to studying the blow-up problem for the scalar equation
\be
{da \over d\tau} = a \, \big( -C_0 a^2 + D_0 a + 1 \big) = P(a),
\qquad \tau >0, 
\ee
with $D_0 \coloneqq C_0 a_0 + a^{-1}_0 (m_0 - 1)$. 
This polynomial $P(a)$ is positive in an interval $(0, a_+)$ containing $a_0$, with $P(a_+)=0$.
As $\tau$ increases, the solution $a=a(\tau)$ increases and near the value $a_+$ the equation behaves like 
${da \over d\tau} \simeq P'(a_+)\, (a_+ - a)$, and the solution is globally defined in $\tau$.


\subsection{Weak solutions to the Euler-Gowdy system} 

We define our notion of initial data and weak solution for the Euler-Gowdy system as follows, by imposing that solutions have a finite energy which we now define. For simplicity in the presentation, the coefficients $a$ and $1/a$ are taken to be locally bounded functions (this is not necessary).

\begin{definition} 
\label{def-weak-Gowdy-data-Euler}
Consider the Euler-Gowdy equations 
 for the five unknowns $\rho, v, U, A, a$, and fix some time $t_0 \neq 0$.
A set of measurable functions $(\brho, \bv, \bara, \barU, \barU_0, \barA, \barA_0)$ defined on $T^1$ 
with $\brho \geq 0$, $v \in (-1, 1)$, and $a >0$ 
is called an \emph{essential Euler-Gowdy initial data set}
if   
\bel{33-intialregu-Euler}  
\aligned
& {\brho \over 1 - \bv^2} \in L^1(T^1), 
\qquad 
\bara, \, {1 \over \bara} \in L^\infty_\loc(T^1), 
\\
&
\barU, \, \barA \in H^1(T^1),
\qquad 
\barU_0, \, \barA_0 \in L^2(T^1). 
\endaligned
\ee 
Furthermore, such a set of functions $(\brho, \bv, \bara, \barU, \barU_0, \barA, \barA_0, \barnu, \barnu_0)$  defined on $T^1$ 
is called a \emph{Euler-Gowdy initial data set} if, in addition, 
\bel{33-weakconstraints-calcul-Euler}
\aligned 
& \barnu_0  = \barU_0 + {\bSigma_0 \over t_0}  
+ t_0 \, \brho \, \frac{k^2 + \bv^2}{1 - \bv^2},
\\
& \barnu_\theta  = \barU_\theta + {\bSigma_1 \over t_0}  
                             - t_0 \, {\brho \over \bara} \, \frac{(1+k^2) \bv}{1 - \bv^2}, 
\endaligned
\ee
in which 
\be 
\aligned 
\bSigma_0 
&\coloneqq
  - t_0 \, \barU_0 + t_0^2 \, \big( \barU_0^2 + \bara^2 \barU_\theta^2 \big) + {1 \over 4} e^{4\barU} (\barA_0^2 + \bara^2 \barA_\theta^2),
\\
\bSigma_1 
&\coloneqq 
- t_0 \, \barU_\theta + 2t_0^2 \, \barU_0 \, \barU_\theta + {1 \over 2} e^{4 \barU} \, \barA_0 \barA_\theta.
\endaligned
\ee
\end{definition}

Observe that, under the conditions given in the above definition, we deduce
\be
\barnu \in W^{1,1}(T^1), \qquad \barnu_0 \in L^1(T^1).
\ee
Recall that we are interested in solutions defined in a time interval $I=[t_0, \tmax)$ included either in $(-\infty, 0)$ or in 
$(0, +\infty)$. 

\begin{definition}
\label{def-weak-Gowdy-solution-Euler} 
Given any interval $I=[t_0, \tmax)$, a set of measurable functions $(\rho, v, a, U, A, \nu)$ defined on $I \times T^1$ and satisfying
$\rho \geq 0$, $v \in (-1, 1)$, and $a>0$ with
\bel{eq-weakPQ-Euler}
\aligned
& {\rho \over 1- v^2} \in L_\loc^\infty(I, L^1(T^1)),
\qquad
a, \, {1 \over a} \in L_\loc^\infty(I \times T^1),
\\
& U, A \in L_\loc^\infty(I, H^1(T^1)),
\qquad
U_t, A_t \in L_\loc^\infty(I, L^2(T^1)),
\endaligned
\ee
is called a \emph{weak solution to the Euler-Gowdy system} if the equations
\eqref{eq:use} and \eqref{evolution-I}--\eqref{constraint-III} hold in the sense of distributions. 
\end{definition}

Under the conditions in the above definition, we can deduce that
the equations \eqref{evolution-II} and \eqref{constraint-III} hold for some function $\nu$ such that
\bel{eq-weaklambda-Euler}
\aligned 
& \nu \in L_\loc^\infty(I, W^{1,1}(T^1)), 
\qquad 
\nu_t \in L_\loc^\infty(I, L^1(T^1)). 
\endaligned
\ee
Furthermore, we are interested in the initial value problem with the following 
prescribed initial condition (understood again in the distributional sense): 
\be
\aligned
\rho(t_0,\cdot) &= \brho, \quad &&&  v(t_0, \cdot) &= \bv, \quad &&&&& a(t_0, \cdot) &= \bara, 
\\ 
U(t_0,\cdot) &= \barU, \quad &&& A(t_0,\cdot) &= \barA, \quad &&&&& \nu(t_0,\cdot) &= \barnu,
\\
U_t(t_0,\cdot) &= \barU_0, \quad &&& A_t(t_0,\cdot) &= \barA_0,  \quad  &&&&& \nu_t(t_0,\cdot) &= \barnu_0.
\endaligned
\ee
The same remark as made in \eqref{eq:3130}--\eqref{eq:31333} applies here.


\subsection{A space-time integrability property for fluids}
\label{sec:42fluid} 

The results in Section~\ref{sec:Gowdy-Volu} are now extended to the Euler-Gowdy system. A major difference is that, in the Gowdy equations, the function $a$ is no longer constant but is defined from the mass density $\rho$ of the fluid. Moreover, we must also take into account terms involving the mass energy density and the velocity and, interestingly, we conclude with a bound on both the fluid and the geometric variables. We focus here on a priori estimates satisfied by weak solutions, while the existence theory will be the subject of \cite{LeFlochLeFloch-2}. 

As already explained, solutions are defined only on an interval of the form $I= [t_0, \tmax)$ where 
$t_0$ can be negative or positive. 
Consider the functional  
\be
\Vcalt(t) 
\coloneqq \abs{t}^{-3/4}  \int_{(T^1)^3} \sqrt{e^{2(\nu-U)} t^2\over a^2} \, dx^2 dx^3 d\theta
=
\abs{t}^{1/4} \int_{T^1} {e^{\nu-U}  \over a} \, d\theta. 
\ee 
We motivate our definition by observing that, using \eqref{constraint-III} and the definitions of $e$, $S$ and $T$,
\bel{eq:numoinsU}
(\nu-U)_t  
=  {\Sigma_0 \over t} + t\rho \, \frac{k^2 + v^2}{1 - v^2}
= t\, a\, e -{1 \over 4t} + t\, a\, S
= t \,  a (e+T) + {a_t \over a} - {1 \over 4t}. 
\ee
It thus follows that 
\be
{d \over dt} \Vcalt(t) = (\sign t) \, \abs{t}^{5/4}\int_{T^1} e^{\nu-U} \, (e+T) \, d\theta, 
\ee
and this motivates us to also propose the following weighted energy functional for the Euler-Gowdy system,
where $\alpha$ is a real parameter fixed later:
\be 
\Ecalt(t) \coloneqq  \abs{t}^{\alpha} \int_{T^1} e^{\nu-U} \, (e+T) \, d\theta.
\ee 

\paragraph*{The integrability property.} 
In view of \eqref{eq:use}, we compute 
\be
\aligned
& \Big( \abs{t}^\alpha e^{\nu-U} ( e + T) \Big)_t + \Big( \abs{t}^\alpha a e^{\nu-U} (f + M) \Big)_\theta 
\\
& = 
\abs{t}^\alpha e^{\nu-U} X
+ (\sign t)\,\abs{t}^{\alpha-1} e^{\nu-U}\biggl(\alpha(e+T) - 2 \, \et - T - (3k^2 + 1){\rho \over 4a}\biggr),
\endaligned
\ee
with 
\be
\aligned
X
&\coloneqq 
(\nu-U)_t \, (e + T)  + (\nu-U)_\theta \, a \, (f + M)   
\\
&= 
\bigg(  t \,  a \, (e+S) - {1 \over 4t} \bigg)\,( e + T) 
- t \, a \, (f + M)^2,
\endaligned
\ee
where we used \eqref{eq:numoinsU} and the following consequence of \eqref{constraint-III} and \eqref{eq:deffandet}:
\bel{eq:numoinsU-theta}
(\nu-U)_\theta
= {\Sigma_1 \over t}  - t \, {\rho \over a} \, \frac{(1+k^2)v}{1 - v^2}
= - t \, (f+M).
\ee

In other words, we have found 
$$
\aligned
& \Big( \abs{t}^\alpha e^{\nu-U} ( e + T) \Big)_t + \Big( a e^{\nu-U} (f + M) \Big)_\theta 
\\
& =
(\sign t)\,\Biggl( \abs{t}^{\alpha+1} e^{\nu-U} a\, \Bigl( (e+T)(e+S)  - (f + M)^2 \Bigr)
\\[-2ex]
& \qquad \qquad \quad
+ \abs{t}^{\alpha-1} e^{\nu-U}\biggl(\Bigl(\alpha-\frac{1}{4}\Bigr)(e+T) - 2 \, \et - T - (3k^2 + 1){\rho \over 4a}\biggr) \Biggr).
\endaligned
$$
Let us point out that for large enough~$\alpha$ the lower-order term (namely, the last line) is positive since $e\geq\et\geq 0$ and $T \geq \rho/a>0$.
This positivity is not essential: without it the lower-order terms could also be
controlled (on any compact interval of time) by the integrand of our functional.

Next, we write
\be
(e+T)(e+S) - (f + M)^2
= (e+f)(e-f) + (e(T+S) - 2fM) + (TS - M^2)
\ee
as a sum of nonnegative terms.
First, all four terms in the following product are squares
$$
\aligned
(e+f)(e-f)
& =
\Bigg(
{1 \over a} \, \big(U_t - \frac{1}{2t} - a U_\theta \big)^2  
               + \frac{e^{4U}}{4 a t^2}\big(A_t - a A_\theta\big)^2 
\Bigg)
\Bigg(
{1 \over a} \, \big(U_t - \frac{1}{2t} + a U_\theta \big)^2  
               + \frac{e^{4U}}{4 a t^2}\big(A_t + a A_\theta\big)^2 
\Bigg) \\
& =
\bigg( {1 \over a} \, \Big(U_t - \frac{1}{2t}\Big)^2 - a U_\theta^2 \bigg)^2
+ \frac{e^{4U}}{2 a^2 t^2}
\bigg( A_t \big(U_t-\frac{1}{2t}\big) - a^2 A_\theta U_\theta \bigg)^2
\\
& \quad
+ \frac{e^{4U}}{2 t^2}
\bigg( \big(U_t - \frac{1}{2t} \big) A_\theta - U_\theta A_t \bigg)^2
+ \frac{e^{8U}}{16 a^2 t^4} \big(A_t^2 - a^2 A_\theta^2\big)^2
\endaligned
$$
Second, $e\pm f\geq 0$ and
$$
T + S \pm 2 M = \frac{\rho(1+k^2)(1\pm v)}{a(1\mp v)} \geq 0
$$
so $e(T+S)\geq\abs{f}\,\abs{2M}\geq 2fM$.
Third, we compute
$$
TS - M^2 = \frac{k^2\rho^2}{a^2}.
$$

Similarly to our analysis of the Gowdy equation, we reach the following conclusion. 
As far as the integrability is concerned, our basic energy bounds show that the shift of $U_t$ by
$\frac{1}{2t}$ can be suppressed, by writing for instance (wth Cauchy-Schwarz inequality) 
$\|  U_t (U_t^2-a^2U_\theta^2) \|_{L^1 (T^1)}^2   
\leq \|  U_t \|_{L^2 (T^1)}  \, 
\|  U_t^2-a^2U_\theta^2 \|_{L^2 (T^1)}$. 

\begin{theorem}[Integrability property for the Euler-Gowdy system] 
\label{theo:higherinte-Euler}
Weak solutions to the Euler-Gowdy system satisfy
\bel{eq-high-int-Euler} 
{1 \over a} \rho, \,
{1 \over a} U_t^2 - a U_\theta^2, \, 
{1 \over a} A_t^2 - a A_\theta^2, \,
{1 \over a} U_t A_t - a U_\theta A_\theta, \,
U_t A_\theta - U_\theta A_t
\in L^2_\loc (I \times T^1).
\ee 
\end{theorem} 


\section{Nonlinear stability of weakly regular Gowdy solutions}
\label{sec:5}

\subsection{Weakly converging sequences of \HoneWoneone\ solutions}

From now on we focus on the case $a=1$ and postpone the analysis of the coupling to the fluid equations. However, we emphasize that all of the results we establish here have an analogue for the Euler-Gowdy system. The arguments of proof are significantly more involved and are the subject of \cite{LeFlochLeFloch-2}. We are interested in analyzing sequences of Gowdy solutions $(P^n, Q^n)$ that remain uniformly bounded in the $H^1$ norm and establishing stability and instability statements. We emphasize that the instability result below 
can be turned into a stability statement once the initial data set is taken into account. (See~Section~6.) 

Recall that a sequence of functions $f^n: T^1 \to \RR$ which are uniformly bounded in the $L^2$ norm 
are said to \textsl{converge weakly in $L^2$} toward some limit $f^\infty$ if, for every function $\varphi \in L^2(T^1)$, 
\be
\int_{T^1} f^n \varphi \, dx \to \int_{T^1} f^\infty \varphi \, dx. 
\ee

\begin{theorem} 
\label{theo-220}
Gowdy solutions determined by unknown coefficients $P, Q, \lambda$ satisfy the following properties: 
\bei 

\item[1.] \emph{Stability property for the essential Gowdy equations.} 
Any sequence of weak solutions $\Pn, \Qn \in L^\infty(I, H^1(T^1))$ to the essential Gowdy equations, 
satisfying  the uniform energy bound
\bel{33-enerbound} 
\aligned
& \limsup_{n \rightarrow + \infty} \| (\Pn, \Qn, \Pn_t, \Qn_t, \Pn_\theta, \Qn_\theta) \|_{L^\infty(I, L^2(T^1))} < + \infty,  
\endaligned
\ee
subconverges, at least weakly in $H^1(I \times T^1)$, 
to a limit $(\Pinfty,\Qinfty)$ which is a solution to the essential Gowdy equations. 

\item[2.] \emph{Instability property for the full set of Gowdy equations.} The property above \textsl{does not hold} for the full set of evolution and constraint equations, that is, under the bound \eqref{33-enerbound} the limits $(\Pinfty,\Qinfty)$ 
need not satisfy the full set of Einstein equations. 
\eei 
\end{theorem}

\begin{proof}[Proof of 1] 
Under the bound \eqref{33-enerbound} and by a standard diagonal argument, we can extract a subsequence of $(\Pn, \Qn)$ such that 
\be
\aligned
& (P^n, Q^n, P_t^n, Q_t^n, P_\theta^n,  Q_\theta^n)(t, \cdot) \rightharpoonup (P^\infty, Q^\infty, P_t^\infty, Q_t^\infty, P_\theta^\infty, Q_\theta^\infty)(t, \cdot)
\\
&  \text{ in the weak $ L^2(T^1)$ topology for all rational  times $t \in I$.}  
\endaligned
\ee 
In addition, the bound on the time derivative implies that this convergence holds at all times:  
\be
\aligned
& (P^n, Q^n, P_\theta^n,  Q_\theta^n)(t, \cdot) \rightharpoonup (P^\infty, Q^\infty, P_\theta^\infty, Q_\theta^\infty)(t, \cdot)
\\
&    \text{ in the weak $ L^2(T^1)$ topology for all times $t \in I$.}  
\endaligned
\ee
We also have 
\be
\aligned
& (P_t^n, Q_t^n) \rightharpoonup (P_t^\infty, Q_t^\infty)
\quad 
\text{ in the weak $L^2(I \times T^1)$ topology.}  
\endaligned
\ee
Next, we observe that our first-order formulation  \eqref{eq:33-8-deux}, that is, 
\bel{eq:33-8-deux-n} 
\aligned
\big( t \, \Omega(R^n,Q^n)^{-2}   R^n_t \big)_t - \big( t \, \Omega(R^n,Q^n)^{-2}   R^n_\theta \big) _\theta = 0, 
\\
\big( t \, \Omega(R^n,Q^n)^{-2} Q^n_t \big)_t - \big( t \, \Omega(R^n,Q^n)^{-2} Q^n_\theta \big)_\theta  = 0, 
\endaligned
\ee  
involves only products of weakly converging functions by strongly converging functions in the $L^2(I \times T^1)$ topology. 
Such products are stable under weak convergence and we obtain 
$$
\aligned
& \Omega(R^n,Q^n) \to \Omega(R^\infty,Q^\infty) && \text{ strongly in $L^2(T^1)$ for all times,} 
\\
& \Omega(R^n,Q^n)^{-2}   R^n_t \rightharpoonup \Omega(R^\infty,Q^\infty)^{-2}   R^\infty_t  \quad 
&& \text{ weakly in $L^2(I \times T^1)$.}  
\endaligned
$$
We thus deduce that 
\bel{eq:33-8-deux-infty} 
\aligned
\big( t \, \Omega(R^\infty,Q^\infty)^{-2}   R^\infty_t \big)_t - \big( t \, \Omega(R^\infty,Q^\infty)^{-2}   R^\infty_\theta \big) _\theta = 0, 
\\
\big( t \, \Omega(R^\infty,Q^\infty)^{-2} Q^\infty_t \big)_t - \big( t \, \Omega(R^\infty,Q^\infty)^{-2} Q^\infty_\theta \big)_\theta  = 0,
\endaligned
\ee   
so that the essential Gowdy equations hold in the limit. 
\end{proof}

\begin{proof}[Proof of 2]
Consider the full set of Einstein equations. 
The sequence of coefficients $\lambda^n$ computed from $(P^n, Q^n)$ is bounded in $W^{1,1}$ and, therefore, converges to a function of bounded variation, which we denote by $\lambda^\infty$, so 
\be
\lambda^n \to \lambda^\infty \qquad \text{ almost everywhere.} 
\ee
From the constraint equations 
\be
\aligned
{1 \over t} \lambda^n_t 
& = (P_t^n)^2 + (P_\theta^n)^2  
            + e^{2P^n} \big( (Q_t^n)^2 + (Q_\theta^n)^2 \big), 
\\
{1 \over t} 
\lambda^n_\theta 
& = 2 P_t^n P_\theta^n + 2 e^{2P^n} Q^n_t Q^n_\theta,
\endaligned
\ee  
we can solely conclude that 
\be
\aligned
{1 \over t} \lambda^\infty_t 
& = \lim_{n \rightarrow + \infty} \Big((P_t^n)^2 + (P_\theta^n)^2  
            + e^{2P^n} \big( (Q_t^n)^2 + (Q_\theta^n)^2 \big) \Big), 
\\
{1 \over t} 
\lambda^\infty_\theta 
& = \lim_{n \rightarrow + \infty} \Big(2 P_t^n P_\theta^n + 2 e^{2P^n} Q^n_t Q^n_\theta \Big). 
\endaligned
\ee   
The right-hand sides of these two equations converge to \textsl{bounded measures} and, in particular need not coincide with 
the expressions computed from $(P^\infty, Q^\infty)$. 
\end{proof}


\subsection{The nonlinear stability of \HoneWoneone\ Gowdy solutions} 

By now requiring that the given initial data set converges strongly in the norm under consideration, we prove that the full set of Einstein equations can be established, which therefore completes the study of Gowdy spacetime in the $H^1$-$W^{1,1}$ regularity class we have proposed in Definitions~\ref{def-weak-Gowdy-data} and \ref{def-weak-Gowdy-solution}.

\begin{theorem}[Well-posedness theory for $H^1$-$W^{1,1}$ regular Gowdy spacetimes]
\label{33-theoG}   
Consider the class of weakly regular Gowdy solutions:  
\bei 

\item \emph{Existence and uniqueness.} 
Given any $H^1$-$W^{1,1}$ initial data set $(\barP_0, \barP, \barQ_0, \barQ, \barlambda_0, \barlambda)$ prescribed at some time $t_0$, there exists a $H^1$-$W^{1,1}$ weak solution $(P,Q, \lambda)$ to the initial value problem for the Gowdy equations, which is unique and is defined for all times $t \in I=(0, + \infty)$.

\item \emph{Nonlinear stability.} 
If $(P, Q, \lambda)$ and $(P',Q', \lambda')$ are two $H^1$--$W^{1,1}$ weak solutions  
 defined for all $t \in I$, then the following inequality holds
\bel{eq:H494}
\aligned
& \dbf_{H^1W^{1,1}}\big( (P,Q,\lambda), (P',Q',\lambda') \big)(t) \lesssim \dbf_{H^1W^{1,1}}\big( (P,Q,\lambda), (P',Q',\lambda') \big)(t_0), \qquad t \in I, 
\endaligned
\ee
in terms of the ``distance''  
\bel{eq:ourdist}
\aligned
\dbf_{H^1W^{1,1}}\big( (P,Q,\lambda), (P',Q',\lambda') \big)(t)
\coloneqq {}
&
\| (P - P') (t), (Q-Q')(t) \|_{H^1(T^1)} 
\\
&
+
\| (P_t-P_t') (t), (Q_t-Q_t') (t) \|_{L^2(T^1)} 
\\
&
+
\| (\lambda - \lambda') (t) \|_{W^{1,1}(T^1)}
+
\| (\lambda_t - \lambda_t') (t) \|_{L^1(T^1)}.
\endaligned
\ee
Here, $t_0 \in I$ is any chosen initial time and the implied constant is uniform on every compact interval of time and depends upon the $H^1$--$W^{1,1}$ norm of the solutions.
\eei 
\end{theorem}

The following corollary is an immediate consequence of \eqref{eq:H494}. The \textsl{weak stability} property
 below should be compared with the \textsl{weak instability} property that we pointed out in Theorem~\ref{theo-220}.

\begin{corollary}[Convergence property for weakly regular Gowdy spacetimes]
\label{cor-facile}
With the notation in Theorem~\ref{33-theoG}, any sequence $(P^n, Q^n, \lambda^n)$ of $H^1$-$W^{1,1}$ weak solutions 
which satisfies the uniform energy bound \eqref{33-enerbound} 
and converges strongly on a hypersurface of constant time $t_0 \in I=(0, + \infty)$, 
enjoys the following properties:  
\bei 

\item The whole sequence $(P^n, Q^n, \lambda^n)$ converges strongly in $H^1$-$W^{1,1}$ at all times $t \in I$ toward some $H^1$-$W^{1,1}$ limiting function $(P^\infty, Q^\infty, \lambda^\infty)$.

\item In particular, the energy density $\lambda^n_t$ and the energy flux density $\lambda_x^n$ converge strongly in $L_\loc^\infty(I, L^1(T^1))$ toward a unique limit $\lambda_t^\infty$ and $\lambda_x^\infty$. 

\item Consequently, $(P^\infty, Q^\infty, \lambda^\infty)$ is a $H^1$-$W^{1,1}$ Gowdy solution. 
\eei
\end{corollary}

\begin{remark} As we will show in Step 2 of the proof of Theorem~\ref{33-theoG}, solutions to the Gowdy equations also enjoy the following additional regularity (for any interval $[t_1, t_2] \subset I$): 
$$
\sup_{\theta_0 \in T^1} \big\| \big(P_t, P_\theta, Q_t, Q_\theta\big) (\cdot, \theta_0) \big\|_{L^2(t_1, t_2)} 
\lesssim \| (P_\theta, P_t, Q_\theta, Q_t ) \|_{L^2(T^1)}. 
$$
\end{remark}


\subsection{Proof of Theorem~\ref{33-theoG}}

{\bf Step 1.} We proceed with a density argument on the initial data and show the existence and uniqueness of the weak solutions by relying on the nonlinear stability statement \eqref{eq:H494} ---which
we will establish next for sufficiently regular solutions. 
\bei 

\item On one hand, from any sufficiently regular initial data set $(\barP, \barQ, \barP_0, \barQ_0)$ prescribed at a time $t=t_0 \in I$, elementary arguments for $1+1$ nonlinear wave equations yield the existence of a classical solution $P,Q$ (say of class $C^2$) defined in small interval of time (at least). 

\item On the other hand, any weakly regular initial data
$\barP, \barQ \in H^1(T^1)$ and $\barP_0, \barQ_0 \in L^2(T^1)$
can be approximated by smooth functions $\barP^n, \barQ^n, \barP_0^n, \barQ_0^n$ such that the following strong convergence property holds (at the initial time only):  
\bel{33-eq-93} 
\aligned
& \barP^n \to \barP, \qquad \barQ^n \to \barQ \quad && \text{ strongly in } H^1(T^1), 
\\
& \barP_0^n \to \barP_0, \qquad \barQ_0^n \to \barQ_0 \quad && \text{ strongly in } L^2(T^1). 
\endaligned
\ee
 \eei
The energy balance law provides us with a uniform control of the energy of the solutions $(P^n, Q^n)$, that is, 
\eqref{33-enerbound}  holds on every compact time interval $[t_1, t_2] \subset I$. 
From this sequence $(P^n, Q^n)$, we extract a subsequence that is weakly converging to some limit denoted by $(P,Q)$:
\bel{33-eq-91} 
\aligned
& P^n(t, \cdot) \to P(t, \cdot), \qquad && Q^n(t, \cdot) \to Q(t, \cdot) 
	\quad && \text{ weakly in } H^1(T^1), 
\\
& P_t^n \to P_t, \qquad && Q_t^n \to Q_t
	\quad && \text{ weakly in } L^2(I \times T^1). 
\endaligned
\ee
We claim that, in fact, 
\bel{33-eq-99} 
\aligned
& P^n(t, \cdot) \to P(t, \cdot), \qquad &&Q^n(t, \cdot) \to Q(t, \cdot) \quad && \text{ strongly in } H^1(T^1), 
\\
& P_t^n (t, \cdot) \to P_t(t, \cdot), \qquad &&Q_t^n(t, \cdot) \to Q_t (t, \cdot) \quad && \text{ strongly in } H^1(T^1), 
\endaligned
\ee
Indeed, this is immediate from the nonlinear stability estimate \eqref{eq:H494}, namely for all $n, n'$ 
$$
\aligned
& \sup_{t \in [t_1, t_2]} 
\| (P^n - P^{n'}, P_t^n - P_t^{n'}, P^n_\theta - P^{n'}_\theta) (t) \|_{L^2(T^1)} 
\\
& \lesssim  
\| (P^n - P^{n'}, P_t^n - P_t^{n'}, P^n_\theta - P^{n'}_\theta) (t_0) \|_{L^2(T^1)},
\endaligned
$$
which shows that the sequence is Cauchy in the $L^2$ norm.  This completes the argument of the existence of the weak solution; the uniqueness follows from the nonlinear stability inequality and the uniqueness of classical solutions. The classical solutions are also known to be defined for the whole range of time and this property carries over to weak solutions.

\vskip.3cm

\noindent{\bf Step 2. The formulation in quadratic variables.} 
For our proof of the nonlinear estimate \eqref{eq:H494}, it is convenient to rely on the quadratic formulation we proposed in Section~\ref{sec:Gowdy-quadra}. 
Given two Gowdy solutions $\Psi=(P_0, P_1, S_0, S_1, \lambda_0, \lambda_1)$ and $\Psi'=(P_0', P_1', S_0', S_1', \lambda_0, \lambda_1)$ expressed in quadratic variables, we focus on the ``essential variables'' $\psi, \psi'$ defined by 
\be
\Psi\eqqcolon(\psi, \lambda_0, \lambda_1), \qquad \Psi'\eqqcolon(\psi', \lambda_0, \lambda_1). 
\ee
It is not difficult to check that what we need to prove is the following inequality in the $L^2$ norm: 
\be 
\| (\psi - \psi') (t) \|_{L^2(T^1)} \lesssim \| (\psi - \psi') (t_0) \|_{L^2(T^1)}.  
\ee
Recall that we already have 
\be
\| \psi(t) \|_{L^2(T^1)} \lesssim \| \psi (t_0) \|_{L^2(T^1)},
\qquad
\| \psi'(t) \|_{L^2(T^1)} \lesssim \| \psi' (t_0) \|_{L^2(T^1)}.
\ee  
We will use later the following observation. 
By extending the solution $\psi$ by periodicity to all $\theta \in \RR$ 
and integrating the second equation of \eqref{eq-quadra2}, namely
\be
\lambda_{1t} - \lambda_{0\theta} = 0,
\ee
over the triangular domain 
\be
\Omega_{\theta_0} \coloneqq \big\{ (t, \theta) \, \bigm/ \, t_0 \leq t \leq  t_1; \,
 \theta_0 \leq \theta \leq  \theta_0 + t_1 - t \big\}
\ee
for any given $\theta_0$ with $\theta_1 \coloneqq \theta_0+(t_1-t_0)$, we obtain 
\bel{eq:timeintpsi0}
\int_{\theta_0}^{\theta_1} \lambda_1(t, \theta_0) \, d\theta 
+ \int_{t_0}^{t_1} (\lambda_0 + \lambda_1)(t, \theta_0 + t_1 - t) dt 
- \int_{t_0}^{t_1} \lambda_0(t, \theta_0) dt = 0. 
\ee 
The first two terms are controlled by the energy of $\psi$, by observing that 
$\abs{\lambda_1} \leq \lambda_0$ and 
by deriving a completely similar identity from the energy equation 
\be
\lambda_{0t} - \lambda_{1\theta} + {\lambda_0 \over t} =  - (P_0)^2 + (P_1)^2  - (S_0)^2 + (S_1)^2 
\ee
over the (larger) domain 
\be
\widetilde \Omega_{\theta_0} \coloneqq \big\{ (t, \theta) \, \bigm/ \, t_0 \leq t \leq  t_1; \,
 \theta_0 - t_1 + t \leq \theta \leq  \theta_0 + t_1 - t \big\}. 
\ee
Consequently, $\theta_0 \mapsto \int_{t_0}^{t_1} \lambda_0(t, \theta_0) \, d\theta$ is also bounded by the initial energy, that is, 
we have proven the $L^2$ time-integral estimates (for both solutions $\psi, \psi'$): 
\bel{eq:timeintpsi}
\sup_{\theta_0 \in T^1} \| \psi(\cdot, \theta_0) \|_{L^2(t_0, t_1))} \lesssim \| \psi (t_0) \|_{L^2(T^1)},
\qquad
\sup_{\theta_0 \in T^1} \| \psi'(\cdot, \theta_0) \|_{L^2(t_0, t_1))} \lesssim \| \psi'(t_0) \|_{L^2(T^1)}.
\ee

\vskip.3cm

\noindent{\bf Step 3. The energy functional for two solutions.} 
From the Gowdy equations, it is straightforward to derive the following energy identity involving the two solutions under consideration:  
\bel{eq-energydiffe}
E(\psi,\psi')_t + F(\psi, \psi')_\theta = -{2 \over t} G(\psi, \psi') + M(\psi, \psi'), 
\ee 
in which 
$$
E(\psi, \psi') \coloneqq E(\psi- \psi'), \qquad F(\psi, \psi') = F(\psi- \psi'), \qquad G(\psi, \psi') = G(\psi- \psi'),
$$ 
and, after a tedious calculation and setting $\Pb_0 \coloneqq (P_0+P_0')/2$, etc., 
\begin{subequations} 
\label{eq:4201} 
\bel{eq:4485}
\aligned
\frac{1}{2} M(\psi, \psi') 
& = (P_0 - P_0') \Big(\Sb_0 (S_0 - S_0') - \Sb_1 (S_1 - S_1')\Big)
+ (P_1 - P_1') \Big(\Sb_1 (S_0 - S_0') - \Sb_0 (S_1 - S_1')\Big)
\\
& \quad
- (S_0 - S_0') \Big(\Pb_0 (S_0 - S_0') - \Pb_1 (S_1 - S_1')\Big)
- (S_1 - S_1') \Big(\Pb_1 (S_0 - S_0') - \Pb_0 (S_1 - S_1')\Big)
.
\endaligned
\ee 
Observe that the cubic nonlinearity $M(\psi, \psi')$ obeys
\bel{eq:Nulltodo}
\aligned
| M(\psi,\psi') | & \lesssim \big( E(\psi- \psi') \big)^{1/2} \bigl( N(\psi, \psi- \psi')^2 + N(\psi', \psi- \psi')^2\bigr)^{1/2}, 
\endaligned
\ee 
\end{subequations}
in which
\bel{eq:nullformnotation-gene}
\aligned
N(\psi, \psi - \psi')^2
\coloneqq {}
&\bigl(P_0(P_0-P_0')-P_1(P_1-P_1')\bigr)^2
+\bigl(P_0(P_1-P_1')-P_1(P_0-P_0')\bigr)^2
\\
&
+\bigl(P_0(S_0-S_0')-P_1(S_1-S_1')\bigr)^2
+\bigl(P_0(S_1-S_1')-P_1(S_0-S_0')\bigr)^2
\\
&
+\bigl(S_0(P_0-P_0')-S_1(P_1-P_1')\bigr)^2
+\bigl(S_0(P_1-P_1')-S_1(P_0-P_0')\bigr)^2
\\
&
+\bigl(S_0(S_0-S_0')-S_1(S_1-S_1')\bigr)^2
+\bigl(S_0(S_1-S_1')-S_1(S_0-S_0')\bigr)^2
.
\endaligned
\ee
This generalizes the notation $N(\psi)^2=N(\Psi)^2=E(\psi)^2-F(\psi)^2$ for null terms defined in \eqref{eq:nullformnotation}, namely we have $N(\psi,\psi)=N(\psi)$.

We restrict attention to any given compact interval $[t_0, t_1] \subset I$
and we integrate \eqref{eq-energydiffe} over the torus $T^1$. The energy norm of  $\psi, \psi'$ is already controlled, 
and therefore with implied constants depending upon $t_0$ and $t_1$ we arrive at
\be
\aligned
\Biggl|{d \over dt} \int_{T^1} E(\psi- \psi') \, d\theta\Biggr|
& \lesssim
\Bigg( \int_{T^1} G(\psi- \psi') \, d\theta \Bigg)
\\
& \quad
+
\Bigg(  \int_{T^1} E(\psi - \psi') \, d\theta \Bigg)^{1/2} 
\Bigg( \int_{T^1} \big( N(\psi, \psi- \psi')^2 + N(\psi', \psi- \psi')^2 \big) \, d\theta \Bigg)^{1/2},
\endaligned
\ee
in which we have used \eqref{eq:4485}-\eqref{eq:Nulltodo} to bound the cubic terms. 
Finally, we set 
\bel{eq-def-ABC}
\aligned
A(t) &\coloneqq \int_{T^1} E(\psi- \psi')(t, \theta)  \, d\theta,   
\\
C(t) &\coloneqq \int_{T^1} \big( N(\psi, \psi- \psi')^2 + N(\psi', \psi- \psi')^2 \big) \, d\theta,
\endaligned
\ee
so that on any compact interval $[t_0, t_1]$ we obtain the inequality
\bel{eq:key-esti} 
\biggl|{d \over dt} A(t)\biggr|
\lesssim  A(t) + C(t), \qquad 
t \in [t_0, t_1]. 
\ee
The implied constant depends upon $t_0, t_1$, and the energy of the solutions. 
We are going to derive variants of the inequality \eqref{eq:key-esti} that incorporate a weight depending upon $\psib=(\psi+\psi')/2$, and this will control the space-time integral of the null expressions $N(\psi, \psi - \psi')$ and $N(\psi', \psi - \psi')$.
At the end of our argument, we will have established that the null terms are controlled by the initial data 
\be
\int_{t_0}^{t_1} C(t) dt \lesssim A(t_0), 
\ee
which in combination with \eqref{eq:key-esti} implies (for some sufficiently large constant $c_* > 0$) 
\be
\aligned
e^{c_* t} A(t)
& \lesssim  e^{c_* t_0} A(t_0) + \int_{t_0}^t C(s) ds
 \lesssim A(t_0),  \qquad 
t \in [t_0, t_1]. 
\endaligned
\ee

\vskip.3cm

\noindent{\bf Step 4. The functional for two solutions.} 
Given $\kappa>0$,
we define a  functional associated with two solutions by
\be
\Vcalt_\kappa(t) \coloneqq \int_{T^1} e^{-\kappa \lambdab} \, d\theta, 
\ee
where $\lambdab \coloneqq (\lambda+\lambda')/2$.
This generalizes the functional $\Vcalt(t)$ of Section~\ref{sec:Gowdy-Volu} that is associated to a single solution.
We also introduce weighted energy and null forms in analogy to~\eqref{eq-def-ABC}:
\bel{eq-def-Akappa}
\aligned
A_\kappa(t) &\coloneqq \int_{T^1} E(\psi- \psi')(t, \theta)  \, e^{-\kappa\lambdab} d\theta,
\\
C_\kappa(t) &\coloneqq \int_{T^1} \big( N(\psi, \psi- \psi')^2 + N(\psi', \psi- \psi')^2 \big)(t, \theta) \, e^{-\kappa\lambdab} d\theta.
\endaligned
\ee
Finally, from \eqref{eq-energydiffe} we obtain 
\bel{eq-energydiffe-lambda}
\aligned
& \Big( E(\psi- \psi')  e^{-\kappa \lambdab} \Big)_t + \Big( F(\psi- \psi')  e^{-\kappa \lambdab} \Big)_\theta 
+{2 \over t} G(\psi - \psi')  e^{-\kappa \lambdab} 
\\
& = M(\psi, \psi')  e^{-\kappa \lambdab} 
- \kappa  \Big( E(\psi- \psi')  \lambdab_t + F(\psi- \psi') \lambdab_\theta \Big) e^{-\kappa \lambdab}
\\
& = \biggl( M(\psi, \psi') 
- \frac{t \kappa}{2}\big( N(\psi, \psi - \psi')^2 + N(\psi', \psi - \psi')^2 \big)
\biggr) \,  e^{-\kappa \lambdab},
\endaligned
\ee 
where we used the observation that
\be
E(\psi- \psi')  \lambda_t + F(\psi- \psi') \lambda_\theta
= t \, N(\psi, \psi - \psi')^2.
\ee

Integrating over $T^1$, then using the inequalities \eqref{eq:4201} satisfied by $M$, gives
$$
\aligned
& \Biggl| \frac{d}{dt} A_\kappa(t) + \frac{t\kappa}{2} C_\kappa(t)
+ {2 \over t}  \int_{T^1} G(\psi - \psi') \, e^{-\kappa \lambdab} d\theta \Biggr|
\\
& 
\leq \int_{T^1} \bigl|M(\psi, \psi')\bigr| \, e^{-\kappa \lambdab} d\theta
\\
& \lesssim
\int_{T^1}
\big( E(\psi- \psi') \big)^{1/2} \big( N(\psi, \psi- \psi')^2 + N(\psi', \psi- \psi')^2 \big)^{1/2}
 \, e^{-\kappa \lambdab} d\theta.
\endaligned
$$
Bounding $0\leq G\leq E$ we conclude
\bel{eq:keykey}
\Big| \frac{d}{dt} A_\kappa(t) + \kappa C_\kappa(t) \Big|
\lesssim
A_\kappa(t) + C_\kappa(t).
\ee
This correctly reduces to \eqref{eq:key-esti} for $\kappa=0$.

\vskip.3cm

\noindent{\bf Step 5. Integration of \eqref{eq:keykey}.} We can now integrate \eqref{eq:keykey} over an interval $[t_0, t]$ and, with a sufficiently large constant $c_*> 0$ (independent of $\kappa$), obtain 
$$
\aligned
e^{c_*t} A_\kappa (t) + \int_{t_0}^t e^{c_*s}  \kappa C_\kappa(s) ds
& \lesssim e^{c_*t_0} A_\kappa (t_0)
+ 
\int_{t_0}^t e^{c_*s} C_\kappa(s)\, ds, 
\endaligned
$$
which on a compact interval of time is equivalent to saying 
\bel{eq:56} 
\aligned
A_\kappa (t) + \kappa \int_{t_0}^t  C_\kappa(s) ds
& \lesssim A_\kappa (t_0) + \int_{t_0}^t C_\kappa(s) ds.
\endaligned
\ee
One consequence is that for large enough $\kappa$
$$
\int_{t_0}^t C_\kappa(s) ds \lesssim A_\kappa (t_0)
$$
Using this back into \eqref{eq:56} we obtain
$$
A_\kappa(t) \lesssim A_\kappa(t_0).
$$
Now, since $\lambdab$ is bounded, $A_\kappa$ and $A$ are equivalent up to $\kappa$-dependent constants.
We conclude that
$$
A_\kappa(t) \lesssim A_\kappa(t_0) , \qquad t\in [t_0,t_1]
$$
with $\kappa$-dependent constants.


\section{Weak convergence and solutions with bounded variation} 
\label{sec:6}
   
\subsection{Gowdy solutions with \texorpdfstring{$H^1$-$BV$}{H\textonesuperior-BV} regularity} 
\label{sec:GowdyH1-BV}

In the present section,  we propose several additional notions of solutions. 
For simplicity in the presentation we focus on the Gowdy equations and postpone to \cite{LeFlochLeFloch-2} the generalization to the Euler-Gowdy system. Our standpoint is as follows: in a first stage, we formulate a broad notion of
solution to the evolution part of the Gowdy equations and we establish an existence and stability result for the initial value problem; in a second stage of our analysis, we can take 
Gowdy's constraint equations into account and determine the regularity or integrability required on the initial data set. 
While this strategy by itself does not lead to a new existence result, it has the definite advantage that it separates the existence and the regularity issues.

Consider the equations \eqref{33-eq2-FULL}-\eqref{33-eq2-FULL-2} with unknowns $P, Q, \lambda$. 
As we are going to show, the concept of solutions presented now arises when considering limits of a sequence of $H^1$-$W^{1,1}$ weak solutions. We use the notation $BV(T^1)$ for the space of functions with bounded variation on $T^1$, that is, functions whose distributional derivative is a bounded measure, and we denote by $M(T^1)$ the space of bounded measures on $T^1$. 
In the following definition, the initial data set and the solutions need not satisfy the constraints~\eqref{33-weakconstraints}, and we revisit the treatment of the equation satisfied by $\lambda$. For convenience, we work here with the \textsl{new unknown}
\be
\mu \coloneqq \lambda + P, 
\ee
which in view of \eqref{33-eq2-FULL-2} satisfies  
\bel{eq:ourmuequa}
\mu_{tt} - \mu_{\theta\theta} + {1 \over t} \, P_t = - P_t^2 + P_\theta^2.
\ee
In the following definition the left-hand side of this equation will involve the first-order derivatives of the measures $\mu_t$ and $\mu_\theta$. 

So, we now study the system 
\bel{33-eq2-FULL-Z}
\aligned 
P_{tt} - P_{\theta\theta} + {1 \over t} \, P_t  
& =  e^{2P} (Q_t^2 - Q_\theta^2), 
\\ 
Q_{tt} - Q_{\theta\theta} + {1 \over t} \, Q_t  
& = - 2 (P_t Q_t - P_\theta Q_\theta)
\\
\mu_{tt} - \mu_{\theta\theta} + {1 \over t} \, P_t 
& = - P_t^2 + P_\theta^2.
\endaligned
\ee

\begin{definition} 
\label{def:H1BV}
Consider Gowdy's evolution equations \eqref{33-eq2-FULL-Z} for the unknown coefficients $P, Q, \mu$. 
\bei 

\item A set of functions $(\barP, \barQ, \barP_0, \barQ_0, \barmu, \barmu_0)$ defined on $T^1$ 
is called a \emph{$H^1$-$BV$ initial data set for Gowdy's evolution equations} if   
\bel{33-intialregu-BV} 
\barP, \barQ \in H^1(T^1), \qquad
\barP_0, \barQ_0 \in L^2(T^1),
 \qquad  \barmu \in BV(T^1), 
\qquad
\barmu_0 \in M(T^1).
\ee  

\item  
A triple of functions $(P,Q, \mu)$ defined on $I \times T^1$ ($I \subset (0, + \infty)$ being an interval) 
and satisfying 
\bel{eq:regu53}
\aligned
&  P,Q \in L_\loc^\infty(I, H^1(T^1)), 
\qquad 
P_t, Q_t \in L_\loc^\infty(I, L^2(T^1)), 
\\
& \mu \in L_\loc^\infty(I, BV(T^1)), 
\qquad 
\mu_t \in L_\loc^\infty(I, M(T^1)), 
\endaligned
\ee
 is called a $H^1$-$BV$ \emph{weak solution to Gowdy's evolution equations} if the evolution equations \eqref{33-eq2-FULL-Z} hold in the sense of distributions. 
\eei
\end{definition}

Initial data are required to hold in the distributional sense only, but from the equations themselves we can deduce that they are achieved in a much stronger sense. (We omit the details.)
The arguments of proof we presented earlier can be generalized to solve the initial value problem at the $H^1$-$BV$ level of regularity. Namely, the treatment of the variables $P, Q$ is the same as before since both functions are in $H^1$. 
On the other hand, it is sufficient to observe that the equation \eqref{eq:ourmuequa} satisfied by $\mu$ has a \textsl{linear principal part} 
while  the right-hand side belongs to $L_\loc^1(I \times T^1)$ and is known once $P,Q,$ have been determined from the first two equations in \eqref{33-eq2-FULL-Z}. By setting $x_\pm \coloneqq (\theta \pm t)/2$ we can rewrite  \eqref{eq:ourmuequa} in the form  
\bel{eq:ourmuequa-2}
\mu_{x_- x_+}  = {P_t \over t} + P_t^2 - P_\theta^2 \in L_\loc^1(I \times T^1),
\ee
which can then be explicitly integrated along characteristics. It thus follows that, given initial data $\barmu, \barmu_0$ at some $t_0 \in I$, the equation \eqref{eq:ourmuequa-2} admits a solution $\mu$ satisfying the regularity \eqref{eq:regu53} and the prescribed initial condition. Of course, by our integrability property, the right-hand side of \eqref{eq:ourmuequa} satisfies $P_t^2 - P_\theta^2 \in L_\loc^2(I \times T^1)$, but we do not need this property in the present argument. 


Our main motivation for introducing the above definition comes from considering limits of Gowdy solutions. In the following statement, we consider a sequence of weak solutions but, obviously, our result also applies to a sequence of smooth solutions. While we already remarked that the Einstein equations \textsl{are not stable under weak convergence}, we can now determine the equations satisfied at the limit. 
The essential Einstein equations are indeed stable, but the constraint equations are \textsl{not stable} 
and a new contribution arises which is determined by solving a wave equation.  

\begin{theorem}[Sequences of weakly converging sequences of Gowdy solutions]  
\label{theo-4400}
Consider a sequence $(\Pn, \Qn, \mu^n)$ of $H^1$-$W^{1,1}$ weak solutions to the full set of Gowdy equations 
satisfying  the energy bound \eqref{33-enerbound}. 
Then, for a subsequence at least, 
$(\Pn, \Qn, \mu^n)$  converges to a limit 
$(P^\infty, Q^\infty, \mu^\infty)$ which is a $H^1$-$BV$ solution
\bel{33-eq2-FULL-limit}
\aligned 
P^\infty_{tt} - P^\infty_{\theta\theta} + {1 \over t} \, P^\infty_t  
& =  e^{2P^\infty}  \big( (Q^\infty_t)^2 - (Q^\infty_\theta)^2 \big), 
\\ 
Q^\infty_{tt} - Q^\infty_{\theta\theta} + {1 \over t} \, Q^\infty_t  
& = - 2 \big( P^\infty_t Q^\infty_t - P^\infty_\theta Q^\infty_\theta \big),
\\
\mu^\infty_{tt} - \mu^\infty_{\theta\theta} + {1 \over t} \, P^\infty_t
& = - (P^\infty_t)^2 + (P^\infty_\theta)^2. 
\endaligned
\ee 
Moreover, Gowdy's constraint equations are no longer satisfied in the limit $n \rightarrow +\infty$ 
and, instead, one has
\bel{33-weakconstraints-limit}
\aligned
{1 \over t} \mu^\infty_t 
& = {1 \over t} P^\infty_t 
+ (P^\infty_t)^2 + (P^\infty_\theta)^2  
            + e^{2P^\infty} \big( (Q^\infty_t)^2 + (Q^\infty_\theta)^2 \big)  + {1 \over t} \phi_t, 
\\
{1 \over t} \mu^\infty_\theta 
& = {1 \over t} P^\infty_\theta 
	+ 2 P^\infty_t P^\infty_\theta + 2 e^{2P^\infty} Q^\infty_t Q^\infty_\theta + {1 \over t} \phi_\theta. 
\endaligned
\ee  
in which the scalar field  
$\phi \in L_\loc^\infty(I, BV(T^1))$
is a solution to the wave equation 
\bel{eq:wavephi}
\phi_{tt} - \phi_{\theta\theta} = 0.
\ee
\end{theorem}
  
\begin{definition} In the context of the above theorem, the scalar field $\phi$ arising in Theorem~ \ref{theo-4400}
will be refered to as the \emph{energy defect field} generated by the given sequence of Gowdy solutions.
\end{definition}

The scalar field $\phi$ represents the ``matter generated'' from a weakly converging sequence of \textsl{vacuum solutions.} 
Theorem~\ref{theo-4400} is consistent with our earlier nonlinear stability theorem which treated a sequence of initial data sets converging in the strong sense: in this case,  the spurious matter terms $\phi$ and $\phi_t$ vanish at the initial time and our wave equation \eqref{eq:wavephi} implies that $\phi$ vanishes identically at all time. 

\begin{proof}
\noindent{\bf Step 1.}  
The Sobolev compactness embedding theorem shows that the sequence $(\Pn, \Qn) \in H^1$  subconverge in the 
$H^1$ norm, while Helly's theorem implies that the sequence $\mu^n \in W^{1,1}$ sub-converges almost everywhere and weakly-star in measure toward some limit $\mu^\infty$.  
We prove next that null forms are stable under weak convergence, which immediately tells us 
that the equations \eqref{33-eq2-FULL-limit} hold in the limit. 

First, we observe that 
\bel{eq:3756}
P^n_{tt} - P^n_{\theta\theta} 
= A^n \coloneqq -  {1 \over t} \, P^n_t +  e^{2P^n} ( (Q^n_t)^2 - (Q^n_\theta)^2), 
\ee
in which $P^n$ is uniformly bounded in the $H^1$ norm while $A^n$ is bounded in $L^1$. So, extracting a subsequence if necessary we can assume that $P^n \rightharpoonup P^\infty$ weakly in $H^1$ while $A^n \to A^\infty$ in the sense of bounded measures. Obviously, we have 
$$
P^\infty_{tt} - P^\infty_{\theta\theta} = A^\infty,
$$
which multiplied by $P^\infty$ gives
\bel{eq:hd49} 
\big(P^\infty P^\infty_t \big)_t - \big(P^\infty P^\infty_\theta \big)_\theta 
= P^\infty A^\infty + (P^\infty_t)^2 - (P^\infty_\theta)^2. 
\ee
On the other hand, we can also multiply \eqref{eq:3756} by $P^n$ and obtain 
$$
\big( P^n P^n_t\big)_t - \big(P^n P^n_\theta \big)_\theta 
= P^n A^n + (P^n_t)^2 - (P^n_\theta)^2. 
$$
Recalling that products of strongly convergent terms by weakly convergent terms are stable, we obtain 
\bel{eq:hd50} 
\big( P^\infty P^\infty_t\big)_t - \big(P^\infty P^\infty_\theta \big)_\theta 
= P^\infty A^\infty + \lim_{n \rightarrow + \infty} \big( (P^n_t)^2 - (P^n_\theta)^2 \big). 
\ee
By comparing \eqref{eq:hd49} and \eqref{eq:hd50}, we conclude that the null form 
$$
\lim_{n \rightarrow + \infty} \big( (P^n_t)^2 - (P^n_\theta)^2 \big) = (P^\infty_t)^2 - (P^\infty_\theta)^2
$$
converges in the sense of distributions. Similarly, multiplying instead by $Q^\infty$ and by $Q^n$ gives a proof that $P_t Q_t - P_\theta Q_\theta$ is stable under weak convergence.
Repeating the argument for
$$
Q^\infty_{tt} - Q^\infty_{\theta\theta} 
=
- {1 \over t} \, Q^\infty_t  - 2 (P^\infty_t Q^\infty_t - P^\infty_\theta Q^\infty_\theta)
$$
shows that $(Q_t)^2 - (Q_\theta)^2$ is also stable under weak convergence.

\vskip.3cm

\noindent{\bf Step 2.} Next, we establish the energy equations for $P^\infty, Q^\infty$,  i.e.
\bel{eq:us}
\aligned 
\Bigg(
(P^\infty_t)^2+ (P^\infty_\theta)^2 + e^{2P^\infty} \big( (Q^\infty_t)^2 + (Q^\infty_\theta)^2 \big) 
\Bigg)_t 
- \big( 2 P^\infty_t P^\infty_\theta  + 2 e^{2P^\infty} Q^\infty_t Q^\infty_\theta \big)_\theta 
& =  - {2 \over  t} \,  \Big( (P^\infty_t)^2  + e^{2P^\infty} (Q^\infty_t)^2 \Big)
\endaligned
\ee
and 
\bel{eq:us-du}
\aligned 
\Bigg( 2 t P^\infty_t P^\infty_\theta  + 2 t e^{2P^\infty} Q^\infty_t Q^\infty_\theta 
\Bigg)_t 
- 
t \Big( 
(P^\infty_t)^2+ (P^\infty_\theta)^2 + e^{2P^\infty} \big( (Q^\infty_t)^2 + (Q^\infty_\theta)^2 \big) 
\Big)_\theta 
& = 0. 
\endaligned
\ee
We observe  
that the right-hand side of the equations satisfied by $(P^\infty, Q^\infty)$ belong to $L_\loc^2$ (in space and time) thanks to our integrability property.
We have 
\bel{33-eL-Z}
\aligned 
P^\infty_{tt} - P^\infty_{\theta\theta} + {1 \over t} \, P^\infty_t  
& =  e^{2P^\infty} \big((Q^\infty_t)^2 - (Q_\theta^\infty)^2\big) \in L^2_\loc(I \times T^1), 
\\ 
Q^\infty_{tt} - Q^\infty_{\theta\theta} + {1 \over t} \, Q^\infty_t  
& = - 2 (P^\infty_t Q^\infty_t - P^\infty_\theta Q^\infty_\theta) \in L^2_\loc(I \times T^1), 
\endaligned
\ee
and we have sufficient regularity for the chain rule to apply. 

\vskip.3cm

\noindent{\bf Step 3.}
In order to analyze the constraint equation, we introduce the two measures 
\be 
\aligned
{1 \over t} \xi_0 
& \coloneqq
- \Big((P_t^\infty)^2 + (P_\theta^\infty)^2 + e^{2P^\infty} \big( (Q_t^\infty)^2 + (Q_\theta^\infty)^2 \big) \Big)
\\
& \quad + 
\lim_{n \rightarrow + \infty} \Big(
(P_t^n)^2 + (P_\theta^n)^2 + e^{2P^\infty} \big( (Q_t^n)^2 + (Q_\theta^n)^2 \big) \Big),
\\
{1 \over t} \xi_1 
& \coloneqq
- 2 \, \Big(P_t^\infty P_\theta^\infty + e^{2P^\infty} Q^\infty_t Q^\infty_\theta \Big)
+ 2 \lim_{n \rightarrow + \infty} \Big( P_t^n P_\theta^n + e^{2P^n} Q^n_t Q^n_\theta \Big). 
\endaligned
\ee
In both right-hand sides, the first term is an integrable function while the second term is solely a measure (obtained as the limit of $L^1$ functions). With this notation, from the constraint equations satisfied by  $(\Pn, \Qn, \mu^n)$ 
we thus have immediately  
$$
\aligned
{1 \over t} \mu^\infty_t 
& = {1 \over t} P^\infty_t 
	+ (P^\infty_t)^2 + (P^\infty_\theta)^2  
            + e^{2P^\infty} \big( (Q^\infty_t)^2 + (Q^\infty_\theta)^2 \big)  + {1 \over t} \xi_0, 
\\
{1 \over t} \mu^\infty_\theta 
& = {1 \over t} P^\infty_\theta 
	+ 2 \big( P^\infty_t P^\infty_\theta + e^{2P^\infty} Q^\infty_t Q^\infty_\theta \big) + {1 \over t} \xi_1. 
\endaligned
$$

Combining the dual energy equation \eqref{eq:us-du} with the compatibility relation $(\mu^\infty_\theta)_t - (\mu^\infty_t)_\theta = 0$ gives
$$
(\xi_1)_t  - (\xi_0)_\theta= 0, 
$$
and therefore there exists a potential, denoted by $\phi$, such that 
\bel{eq:phitphitheta}
\xi_0 = \phi_t, \qquad \xi_1 = \phi_\theta. 
\ee
This function has bounded variation in space and Einstein's contraint equations hold in the generalized form \eqref{33-weakconstraints-limit}. 
In turn, the energy equation \eqref{eq:us} translates to
$$
\big( \mu^\infty_t - \xi_0 \big)_t
 - \big( \mu^\infty_{\theta} - \xi_1 \big)_\theta 
+ {1 \over t} \, P^\infty_t
 = - (P^\infty_t)^2 + (P^\infty_\theta)^2.
$$
Hence, by comparing with the third equation in \eqref{33-eq2-FULL-limit} and using \eqref{eq:phitphitheta}, we find the wave equation \eqref{eq:wavephi} for $\phi$. 
\end{proof}


\subsection{Gowdy solutions with bounded variation}
\label{sec:GowdyBV} 

This section provides an even weaker notion of solution to the essential Gowdy equations in $P,Q$.
Now, we assume that both functions may be discontinuous and we develop a novel theory of weak solutions in the class of functions with bounded variation ---the weakest class of regularity in which these equations make sense.
We recall that bounded variation solutions to the Einstein equations were first constructed in spherical symmetry \cite{Christo-1995} and plane symmetry~\cite{Barnes-2004}. 

We thus consider the system \eqref{33-eq2-FULL-Z}, which we put in our conservative form \eqref{eq:33-8}, namely
\bel{eq:33-8-deux-more} 
\aligned
\Big( t \big( P_t - e^{2P} Q Q_t \big) \Big)_t - \Big( t \big( P_\theta - e^{2P} Q Q_\theta \big) \Big) _\theta 
& = 0, 
\\
\big( t e^{2P}  Q_t \big)_t - \big( t e^{2P} Q_\theta \big)_\theta  
& = 0. 
\endaligned
\ee 
We use Volpert's product \cite{Volpert-1967} (see \eqref{eq:V2} below) in order to give a meaning to the products of a BV function by a measure.

\begin{definition} 
\label{def:BVBV}
Consider the essential Gowdy equations with unknowns $P, Q$:
\bei 

\item A set of functions $(\barP, \barQ, \barP_0, \barQ_0)$
defined on $T^1$ 
is called a \emph{$BV$ initial data set} for the essential Gowdy equations if   
\bel{33-intialregu-BVBV} 
\aligned
& \barP, \barQ
\in BV(T^1), 
\qquad
\barP_0, \barQ_0
\in M(T^1). 
\endaligned
\ee  

\item A pair of functions $(P,Q)$
defined on $I \times T^1$ ($I$ being an interval) satisfying
\bel{eq:6255}
\aligned 
& P,Q
\in L_\loc^\infty(I, BV(T^1)), 
\qquad 
P_t, Q_t
\in L_\loc^\infty(I, M(T^1)), 
\endaligned
\ee
 is called a \emph{BV solution to the essential Gowdy equations} if \eqref{eq:33-8-deux-more} hold
 in the sense of distributions, in which the products 
$e^{2P} Q Q_t$, 
$e^{2P} Q Q_\theta$, 
$e^{2P} Q_t$, and $e^{2P} Q_\theta$ are in the sense of Volpert's product (cf.~Definition~\ref{def-666} below). 

\eei
\end{definition}

As before, a solution assumes the prescribed initial data set at some time $t_0\in I$ if in the sense of distributions 
\be
\aligned
&
P(t_0,\cdot) = \barP, \quad P_t(t_0,\cdot) = \barP_0, 
\quad
Q(t_0,\cdot) = \barQ, \quad Q_t(t_0,\cdot) = \barQ_0,
\endaligned
\ee

\begin{theorem}[Theory of BV solutions to the essential Einstein evolution equations in Gowdy symmetry] 
\label{theo-stableGowdy-44}
Consider the class of Gowdy solutions:   
\bei 

\item[\bf 1.] \emph{Existence of BV solutions.}  
Given any $BV$ initial data set $(\barP_0, \barP, \barQ_0, \barQ)$
prescribed at some $t_0> 0$, there exists a $BV$ solution $(P,Q)$
to the essential Gowdy equations which is defined on an interval $(\tmin, \tmax)$  containing $t_0$, 
until the solution $P$ possibly blows up to $-\infty$ at critical times $\tmin \in [0, t_0)$ or/and $\tmax \in (t_0, +\infty]$. 

\item[\bf 2.] \emph{Stability of solutions.} 
If $(P, Q)$ and $(P',Q')$ are two $BV$ solutions  to the essential Gowdy equations 
 defined for all $t \in I$ and $\theta \in T^1$, then\footnote{Due to the very weak regularity under consideration, only the $L^1$ norm can be controlled, rather than the BV norm (which is not continuous in presence of jump discontinuities).} 
\bel{eq:H494-BV}
\dbf_{BV \text{reg}} \big[P, Q; P', Q'\big](t) 
\lesssim \dbf_{BV \text{reg}} \big[P, Q; P', Q'\big](t_0)
\ee
where 
\be
\aligned 
&\dbf_{BV \text{reg}} \big[P, Q; P', Q'\big](t) 
\\
& \coloneqq 
\| (P - P') (t)\|_{L^1(T^1)} + \| (Q - Q') (t) \|_{L^1(T^1)} 
\\
& \quad +
\| (\Psharp - {\Psharp}') (t)\|_{L^1(T^1)} + \| (\Qsharp - {\Qsharp}') (t) \|_{L^1(T^1)}.
\endaligned
\ee
Here, $t_0 \in I$ is any chosen initial time and the implied constant is uniform on every compact interval of time and depends upon the BV norm of the solutions, while the notation $\Psharp$ (and similarly for the other functions) stands for 
$\Psharp(t,\theta) \coloneqq \int_0^\theta P_t(t,y) \, dy$. 

\eei  
\end{theorem}
 
Before proceeding with the analysis in the BV class, we transform the essential Gowdy equations to a new first-order form.
This allows us to make an interesting contact with the theory of first-order hyperbolic systems in nonconservative form.

Consider the system \eqref{eq:33-8-deux-more}. Recall our first-order formulation \eqref{eq:33-9} of these equations
in terms of $R \coloneqq e^{-2P} + Q^2$ and of $\Omega^2 \coloneqq R - Q^2 = e^{-2P} > 0$, 
\bel{eq:620}
\aligned
\big(t \Omega^{-2} R_t \big)_t - \big( t \Omega^{-2} R_\theta \big)_\theta = 0, 
\\
\bigg(t \Omega^{-2} Q_t \big)_t - \big( t \Omega^{-2} Q_\theta \big)_\theta = 0, 
\endaligned
\ee
together with the algebraic constraint $R-Q^2 > 0$. 

We observe that $\int_{T^1} t \Omega^{-2} R_t d\theta \eqqcolon a$ and $\int_{T^1} t \Omega^{-2} Q_t d\theta \eqqcolon b$ are constants in time
and we define the functions (that are not periodic, namely $\theta\in\RR$ rather than $\theta\in T^1$)
\be
\aligned
A & \coloneqq \int_0^\theta t \Omega^{-2} R_t \, d\theta + \int_{t_0}^t s\Omega^{-2} R_\theta(s,0)\,ds, \\
B & \coloneqq \int_0^\theta t \Omega^{-2} Q_t \, d\theta + \int_{t_0}^t s\Omega^{-2} Q_\theta(s,0)\,ds.
\endaligned
\ee
We find
\bel{eq-formu-84}
\aligned
t \Omega(R, Q)^{-2}  \, R_t - A_\theta & =0, 
\\
A_t - t \Omega(R, Q)^{-2} \, R_\theta &=0, 
\\
t \Omega(R, Q)^{-2} \, Q_t - B_\theta &=0, 
\\
B_t - t \Omega(R, Q)^{-2} \, Q_\theta &=0, \hskip1.cm \Omega^2 = R - Q^2>0.
\endaligned
\ee
The system \eqref{eq-formu-84} admits constant wave speeds $\pm 1$ (with double multiplicity), together with a full basis of eigenvectors, as is clear from
\be
\begin{pmatrix} 
R \\ A \\ Q \\ B
\end{pmatrix}_t 
- 
\begin{pmatrix} 
0 & t^{-1} \Omega^2  & 0  &0 
\\
t \Omega^{-2} & 0  & 0  &0 
\\
0  & 0  &0 & t^{-1} \Omega^2
\\
0 & 0  & t \Omega^{-2}  &0 
\\
\end{pmatrix} 
\begin{pmatrix} 
R \\ A \\ Q \\ B
\end{pmatrix}_\theta
= 
\begin{pmatrix} 
0 \\ 0 \\ 0 \\ 0
\end{pmatrix}. 
\ee
This first-order hyperbolic system is linearly degenerate in the sense of Lax \cite{Lax-1957, Lax-1971}. Furthermore, it is genuinely non-conservative in the sense that it cannot be transformed to a system of four conservation laws of the form $f(R,A,Q,B)_t + h(R,A,Q,B)_\theta =0$.

In \eqref{eq-formu-84}, it is necessary to restrict attention to the range of $R,Q$ for which the coefficients $\Omega(R, Q)^2$ and  $\Omega(R, Q)^{-2}$ remains bounded and non-vanishing. This requires us to guarantee, in our existence argument, that $R-Q^2$ remains bounded below away from zero. This is precisely the restriction on $P$ that we have in Theorem~\ref{theo-stableGowdy-44}.

We now proceed with the analysis of the Cauchy problem for the system \eqref{eq-formu-84}. As mentioned earlier, we use the notion of Volpert's product to define our notion of weak solution. 
We are interested in functions that, for all times, have bounded variation in space; such functions admit at most countably many points of jump discontinuity. Given two functions $u=u(t,\theta)$ and $v=v(t,\theta)$ and a smooth function $g=g(u)$, we define the product of the BV function $g(u)$ by the measure $\del_\theta v$ (and similarly for $\del_t v$) as
the measure $\widehat{g(u)} \del_\theta v$
obtained by extending the definition of $g(u)$, for each time $t$, to a function defined at \textsl{every point} $\theta$ and specifically we set  
\bel{eq:V2}
\widehat{g(u)} (t,\theta) 
 \coloneqq \int_0^1 g\big( z \, u_- (t,\theta) + (1-z) u_+(t,\theta) \big) \, dz, 
\ee
in which $u_\pm(t,\theta)$ denote the left- and right-hand limit of the function $\theta \mapsto u(t, \theta)$. 

Let us state our definition of weak solution.  

\begin{definition} 
\label{def-666}
A set of functions $R,Q,A,B$  with locally bounded variation in space (defined for $t \in I \subset (0,\infty)$ and $\theta \in T^2 \simeq [0,1]$) 
is called a \emph{weak solution with bounded variation} to the first-order version of 
the essential Gowdy equations if \eqref{eq-formu-84} hold as equalities between locally bounded measures, in which the products $\widehat{{\Omega(R, Q)^{-2}}}  \, R_t$, $\widehat{\Omega(R, Q)^{-2}}  \, R_\theta$, 
$\widehat{\Omega(R, Q)^{-2}}  \, Q_t$, and $\widehat{\Omega(R, Q)^{-2}}  \, Q_\theta$, 
are understood as Volpert's products. 
\end{definition}

Then, we establish Theorem~\ref{theo-stableGowdy-44} as a direct corollary of Glimm's theorem \cite{Glimm-1965} 
(providing uniform BV bound for a sequence of approximate solutions to the Cauchy problem) 
and P. LeFloch--T.-P.~Liu's theorem \cite{LeFlochLiu}. The latter uses in an essential way the 
property of pointwise convergence enjoyed by the Glimm scheme and proves that Glimm's approximate solutions converge 
strongly to a weak solution. We observe that these standard arguments assume that the initial data have sufficiently small total variation. However, since our system has constant wave speed, it is not difficult to revisit the total variation estimate and $L^1$ stability argument (for instance, presented in \cite{LeFloch-book} )
in order to establish these properties without any restriction on the size of the total variation of the initial data (and solutions). 


\appendix 

\section{Derivation of the Euler-Gowdy system} 
\label{sec:7}

\subsection*{The Gowdy solutions}  
 
Let us discuss first vacuum solutions \cite{Gowdy-1974,Chrusciel-1990}. 
We are interested in spacetimes $(M,g)$ with $T^2$ symmetry on $T^3$, that is, 
$(3+1)$-dimensional Lorentzian manifolds with topology $I \times T^3$ (where $I$ is an interval) admitting the Lie group $T^2$ as an isometry group acting on the spatial leaves $T^3$. In other words, there exist two linearly independent vector fields with closed orbits $X,Y$ satisfying: 
\be
\aligned 
& \text{Commuting property:}    &&[X,Y] = 0, 
\\
& \text{Spacelike property:}    &&g(X,X) > 0 \text{ and }   g(Y,Y) > 0, 
\\
& \text{Killing conditions:}    && \Lcal_X g = \Lcal_Y g = 0.
\endaligned
\ee
Einstein vacuum equations imply that the so-called twists\footnote{Here and below, all Greek indices $\alpha, \beta$ vary from $0$ to $3$.}
$
\Theta_X \coloneqq \eps_{\alpha\beta\gamma\delta} X^\alpha Y^\beta \nabla^\gamma X^\delta
$
and $
\Theta_Y \coloneqq \eps_{\alpha\beta\gamma\delta} X^\alpha Y^\beta \nabla^\gamma Y^\delta
$
are \textsl{constants}. By making a suitable linear combination of the given Killing fields, one can always normalize  
one of twists to vanish identically, say $\Theta_Y=0$. 
Gowdy \cite{Gowdy-1974} first studied such spacetimes under the additional requirement that $X, Y$ have vanishing twist constants. This latter condition can be interpreted as saying that the $T^2$-action is orthogonally transitive in the sense that the distribution of $2$-planes of covectors 
$\big( \text{Vect}(X,Y)\big)^\perp$  
is Frobenius integrable. This completes the description of the class of Gowdy spacetimes. 

The equivalence above is easily checked as follows. 
Choose two vectors $Z,T$ orthogonal to $X,Y$ normalized such that
$
\eps^{\alpha\beta\gamma\delta} X_\gamma Y_\delta = Z^\alpha T^\beta - T^\alpha Z^\beta, 
$
where $\eps$ is the canonical volume form $|g|^{1/2} dx^0 \wedge \ldots \wedge dx^n$ in coordinates.   
Then, the distribution of covectors $g(X, \cdot), \, g(Y, \cdot)$ is Frobenius integrable if and only if 
$Z,T, [Z,T]$ are linearly dependent,
that is 
$\eps_{\alpha\beta\gamma\delta} Z^\alpha T^\beta [Z,T]^\gamma =0$.
This is indeed equivalent to the vanishing twist condition $\Theta_X= \Theta_Y =0$. 

We always tacitly assume that the spacetimes under consideration are not flat, that is, do not have identically vanishing curvature. We then express the metric $g$ in the so-called areal coordinates denoted by $(t,\theta,x^2,x^3)$, i.e., the time function $t$ is chosen to coincide with the area $R$ of the orbits of $T^2$-symmetry. 
The existence of this time function is guaranteed (locally at least) by observing that the vector $\nabla R$ is necessarily timelike,  so that the slices of constant time $t=R$ are spacelike hypersurfaces. 
More precisely, this is true and this choice of coordinates is possible if only if the spacetime is not flat, as assumed above. The metric in areal coordinates reads 
\bel{33-metric-g} 
g = e^{\lambda/2} t^{-1/2} \, (- dt^2 + d\theta^2) + t e^P (dx^2 + Q \, dx^3)^2 + e^{-P} \, t \, (dx^3)^2 
\ee
and 
is determined by three metric coefficients $P,Q, \lambda$, which depend on the two independent variables $t,\theta$. The time variable $t$ belongs to an interval $I \subset \RR$ and the variables $\theta,x^2,x^3$ range over the torus $T^1 \simeq [0,1]$ (with periodic boundary conditions).
The coordinates $t, \theta$ parameterize the quotient manifold $M/T^2$. The vectors  
$\del/\del x^2$ and $\del/ \del x^3$ are the Killing fields of the spacetime
and $x^2, x^3$ are coordinates on the toroidal orbits of symmetry $T^2$. 
By construction, the square of the area of the two-dimensional spacelike orbits of symmetry is 
$$
\mbox{det} \left( \begin{matrix}
t e^P & t e^P Q
\\
t e^P Q & t e^P Q^2 + t e^{-P}
\end{matrix}\right)
= t^2.
$$


\subsection*{The Euler-Gowdy system}

Consider next the evolution of a compressible fluid in Gowdy symmetry \cite{LeFlochRendall-2011}. We are interested in four-dimensional Lorentzian manifolds $(M, g)$ satisfying the Einstein equations $G_{\alpha\beta} = T_{\alpha\beta}$, 
where $T_{\alpha\beta}$ denotes the stress-energy tensor of the fluid and $G_{\alpha \beta} \coloneqq R_{\alpha\beta} - (R/2) g_{\alpha\beta}$ denotes the Einstein curvature tensor describing the geometry of the spacetime. Here, $R_{\alpha\beta}$ and $R$ denote the Ricci and scalar curvatures, which must for our purposes be understood in a weak sense~\cite{LeFlochMardare-2007}. 

The energy-momentum tensor of a perfect fluid is
$T_{\alpha\beta} = (\mu + p) \, u_\alpha u_\beta + p \, g_{\alpha\beta}$,
where $\mu \geq 0$ is the rest mass-energy density of the fluid and $u^\alpha$ denotes 
its unit, timelike velocity vector. We assume the linear equation of state 
$p=k^2 \mu$, where $k \in [0,1]$ represents the sound speed in the fluid and does not exceed the speed of light (normalized to unity). The Bianchi identities for the geometry imply  
$\nabla_\alpha T_\beta^\alpha=0$, which are the Euler equations describing the evolution of the fluid. 

We again assume $T^2$ symmetry on $T^3$. 
We recall that, while the ``twist constants'' are indeed constant for $T^2$-symmetric vacuum spacetimes, they are no longer 
constant in $T^2$-symmetric matter spacetimes, but satisfy evolution equations of their own.
Our interest here is in Gowdy-symmetric spacetimes which, by definition, have vanishing ``twist constants''.

In order to formulate the initial value problem for the Euler-Gowdy system, we prescribe an initial data set on a spacelike hypersurface $\Hcal$ whose topology is the 3-torus $T^3$. The initial data are said to be Gowdy-symmetric if they are  invariant under the action of the Lie group $T^2$ and have vanishing twist constants; cf.~the previous subsection.  
We consider spacetimes $(M, g)$ that admit a foliation by a time function $t: M \to  I \in \R$, where $I$ is an interval. More precisely, $M = \bigcup_{t \in I} \Hcal_t$, 
where each $\Hcal_t$ is a compact spacelike Cauchy hypersurface diffeomorphic to the initial hypersurface $\Hcal$ and $g^{\alpha\beta} \del_\alpha t$ is a future-oriented timelike vector field.  In Gowdy symmetry, it is natural to foliate the spacetime by the area function. Its gradient $\nabla t$ is a timelike vector field, so this is always possible as already pointed out in the previous subsection, except if the spacetime is vacuum and flat. 
Either $\nabla t$ or $-\nabla t$ can be chosen to determine the time-orientation of $M$. Since discontinuous solutions of Euler equations are time-irreversible, the spacetime is (uniquely) defined only in the \textsl{future} of the initial hypersurface. 

In areal coordinates, the metric reads 
\be
g = e^{2(\eta-U)}(-a^2d t^2 + d\theta^2) +e^{2U}(dx + Ady)^2+e^{-2U}t^2dy^2, 
\ee
where the four metric coefficients $a,\eta, U, A$ (with $a> 0$) depend upon the time variable $t$ and the space variable $\theta \in T^1\simeq [0,1]$ (with periodic boundary conditions). 

Einstein's constraint equations in Gowdy symmetry imply that the velocity field~$u$ is orthogonal to the two Killing fields,
namely $u=u^0\del_t+u^1\del_\theta$. Since the velocity is normalized by $g(u,u)=-1$, we parametrize it with a scalar $v$ such that
\be
u^0 = {1 \over a} e^{-\eta + U} \big( 1 - v^2 \big)^{-1/2}, 
\qquad 
u^1 = v \, e^{-\eta + U} \big( 1 - v^2 \big)^{-1/2}. 
\ee
It is convenient also to define $\nu$ and $\rho$ as rescalings of $\eta$ and $\mu$:
\be
\nu \coloneqq \eta + \ln a, \qquad 
\rho\coloneqq e^{2(\nu-U)}\mu = a^2 e^{2(\eta-U)} \mu.
\ee
After some tedious calculations \cite{LeFlochRendall-2011} (see also the summary given in \cite[Section 2.1]{GrubicLeFloch-2015}), we find 
that $\nabla^\alpha T_{\alpha\beta} =0$ gives the two Euler equations \eqref{fluid-II} while the evolution and constraint equations for the geometry are \eqref{evolution-I}--\eqref{constraint-III}.


\paragraph*{Acknowledgments.}
Both authors gratefully acknowledge support from the Simons Center for Geometry and Physics, Stony Brook University, at which some of the research for this paper was performed.
The second author (PLF) was partially supported by the Innovative Training Network (ITN) grant 642768 (entitled ModCompShock) and by the Centre National de la Recherche Scientifique (CNRS). This paper was completed when the second author was a visiting fellow at the Courant Institute for Mathematical Sciences, New York University.  
 


\begin{thebibliography}{99}

\bibitem{Andreasson-1999} 
\auth{Andr\'easson H.,}
Global foliations of matter spacetimes with Gowdy symmetry, 
\emph{Commun. Math. Phys.} 206 (1999), 337--366. 

\bibitem{Barnes-2004} 
\auth{Barnes A.P., LeFloch P.G., Schmidt B.G., and Stewart J.M.,}
The Glimm scheme for perfect fluids on plane-symmetric Gowdy spacetimes, 
\emph{Class. Quantum Grav.} 21 (2004), 5043--5074.

\bibitem{Christo-1992} 
\auth{Christodoulou D.,}
Bounded variation solutions of the spherically symmetric Einstein-scalar field equations,  
\emph{Comm. Pure Appl. Math.} 46 (1992), 1131--1220.

\bibitem{Christo-1995} 
\auth{Christodoulou D.,}
Self-gravitating relativistic fluids: a two-phase model,
\emph{Arch. Rational Mech. Anal.} 130 (1995), 343--400. 

\bibitem{Chrusciel-1990}
\auth{Chru\'sciel P.T.,} 
On spacetimes with $U(1) \times U(1)$-symmetric compact Cauchy surfaces,
\emph{Ann. Phys.} 202 (1990), 100--150.

\bibitem{Dafermos-1979} 
\auth{Dafermos C.M.,} 
Hyperbolic balance laws in continuum physics, 
in ``Nonlinear problems in theoretical physics'' (Proc. IX G.I.F.T. Internat. Sem. Theoret. Phys., 
Univ. Zaragoza, Jaca, 1978), 
Lecture Notes in Phys., 98, Springer, Berlin-New York, 1979, pp.~107--121.

\bibitem{Dafermos-book} 
\auth{Dafermos C.M.,}
\textsl{Hyperbolic conservation laws in continuum physics,}
Grundlehren der Mathematischen Wissenschaften, Vol.~325, Fourth edition, 
Springer Verlag, Berlin, 2016. 

\bibitem{EG} 
\auth{Evans L.C. and Gariepy R..F,}
\emph{Measure theory and fine properties of functions,}
C.R.C Press, Chapman and Hall/CRC, 2015. 

\bibitem{Glimm-1965} 
\auth{Glimm J.,}  
Solutions in the large for nonlinear hyperbolic systems of equations, 
\emph{Comm. Pure Appl. Math.} 18 (1965), 697--715.

\bibitem{Gowdy-1974} 
\auth{Gowdy R.,} 
Vacuum spacetimes with two-parameter spacelike isometry groups and compact invariant hypersurfaces: 
topologies and boundary conditions, 
\emph{Ann. Phys.} 83 (1974), 203--241. 

\bibitem{GrubicLeFloch-2013} 
\auth{Grubic N. and LeFloch P.G.,}  
Weakly regular Euler-Gowdy spacetimes with Gowdy symmetry. The global areal foliation, 
\emph{Arch. Rational Mech. Anal.} 2008 (2013), 391--428. 

\bibitem{GrubicLeFloch-2015}  
\auth{Grubic N. and LeFloch P.G.,}    
On the area of the symmetry orbits in weakly regular Euler-Gowdy spacetimes with Gowdy symmetry, 
\emph{SIAM J. Math. Anal.} 47 (2015), 669--683.

\bibitem{Lax-1957} 
\auth{Lax P.D.,}
Hyperbolic systems of conservation laws. II, 
\emph{Comm. Pure Appl. Math.} 10 (1957), 537--566.

\bibitem{Lax-1971} 
\auth{Lax P.D.,}
Shock waves and entropy, 
in ``Contributions to Nonlinear Functional Analysis'', ed.~E. Zarantonello, 
Academic Press, New York, 1971, pp.~603--634.

\bibitem{LeFlochLeFloch-2} 
\auth{Le Floch B. and LeFloch P.G.,} 
On the global evolution of self-gravitating matter.
Nonlinear interactions in $T^2$ symmetry, Preprint, 2018.  

\bibitem{LeFlochLeFloch-3} 
\auth{Le Floch B. and LeFloch P.G.,} 
On the global evolution of self-gravitating matter.
Nonlinear interactions in spherical symmetry, in preparation.  

\bibitem{LeFloch-book} 
\auth{LeFloch P.G.,} 
\textsl{Hyperbolic systems of conservation laws. The theory of classical and nonclassical shock waves,} 
Lectures in Mathematics, ETH Z\"urich, Birkh\"auser, 2002. 

\bibitem{LeFlochLiu}
\auth{LeFloch P.G. and Liu T.-P.,} 
Existence theory for nonlinear hyperbolic systems in nonconservative form, 
\emph{Forum Math.} 5 (1993), 261--280.

\bibitem{LeFlochMardare-2007} 
\auth{LeFloch P.G. and Mardare C.,}
Definition and weak stability of spacetimes with distributional curvature, 
\emph{Portugal Math.} 64 (2007), 535--573.

\bibitem{LeFlochRendall-2011} 
\auth{LeFloch P.G. and Rendall A.D.,}
A global foliation of Euler-Gowdy spacetimes with Gowdy-symmetry on $T^3$, 
\emph{Arch. Rational Mech. Anal.} 201 (2011), 841--870.  

\bibitem{LeFlochSmulevici-2016}  
\auth{LeFloch P.G. and Smulevici J.,}  
Weakly regular T2-symmetric spacetimes. The future causal geometry of Gowdy spaces, 
\emph{J. Differ. Equa. 260} (2016), 1496--1521.  

\bibitem{LeFlochStewart-2011} 
\auth{LeFloch P.G. and Stewart J.M.,} 
The characteristic initial value problem for plane-symmetric spacetimes with weak regularity, 
\emph{Class. Quantum Grav.} 28 (2011), 145019--145035. 
 
\bibitem{MakinoUkai-1995} 
\auth{Makino T. and Ukai S.,} 
Local smooth solutions of the relativistic Euler equation,
\emph{J. Math. Kyoto Univ.} 35 (1995), 105--114.

\bibitem{Moncrief-1981} 
\auth{Moncrief V.,}
Global properties of Gowdy spacetimes with T3xR topology, 
\emph{Ann. Phys. 132} (1981), 87--107. 

\bibitem{Rendall-book} 
\auth{Rendall A.D.,} 
\textsl{Partial differential equations in general relativity,}
Oxford University Press, Oxford, 2008.

\bibitem{Ringstrom-2004} 
\auth{Ringstrom H.,}  
On a wave map equation arising in general relativity, 
\emph{Comm. Pure Appl. Math.} 57 (2004), 657--703.

\bibitem{SmollerTemple-1993} 
\auth{Smoller J. and Temple B.,} 
 Global solutions of the relativistic Euler equations,
\emph{Comm. Math. Phys.} 156 (1993), 67–-99. 

\bibitem{Volpert-1967}
\auth{Volpert A. I.,}
The space BV and quasilinear equations, 
\emph{USSR Sbornik} 2 (1967), 225--267.

\end{thebibliography}
\end{document}